\documentclass[letterpaper,twocolumn,10pt]{article}
\usepackage{usenix}

\usepackage{cite}
\usepackage{amsmath,amssymb,amsfonts,amsthm}
\usepackage{algorithm}
\usepackage{algpseudocode}
\usepackage{graphicx}
\usepackage{textcomp}
\usepackage{xcolor}
\usepackage{enumitem}
\usepackage{url}
\usepackage{multirow}
\usepackage{float}
\usepackage[font=footnotesize]{caption}
\usepackage{subcaption}

\newtheorem{lemma}{Lemma}
\newtheorem{theorem}{Theorem}
\newtheorem{proposition}{Proposition}

\theoremstyle{definition}
\newtheorem{definition}{Definition}

\usepackage{etoolbox}

\def\BibTeX{{\rm B\kern-.05em{\sc i\kern-.025em b}\kern-.08em
    T\kern-.1667em\lower.7ex\hbox{E}\kern-.125emX}}
\begin{document}

\title{
\textcolor{red}{This is an early publication of a work in progress.}\\
Wrangling Entropy: Next-Generation Multi-Factor Key Derivation, Credential Hashing, and Credential Generation Functions}


\author{
    Colin Roberts\\
    \textit{Multifactor}
    \\colin@multifactor.com
\and
    Vivek Nair\\
    \textit{Multifactor}
    \\vivek@multifactor.com
\and
    Dawn Song\\
    \textit{UC Berkeley}
    \\dawnsong@berkeley.edu
}

\maketitle

\begin{abstract}
The Multi-Factor Key Derivation Function (MFKDF) offered a novel solution to the classic problem of usable client-side key management by incorporating multiple popular authentication factors into a key derivation process, but was later shown to be vulnerable to cryptanalysis that degraded its security over multiple invocations.
In this paper, we present the Entropy State Transition Modeling Framework (ESTMF), a novel cryptanalytic technique designed to reveal pernicious leaks of entropy across multiple invocations of a cryptographic key derivation or hash function, and show that it can be used to correctly identify each of the known vulnerabilities in the original MFKDF construction.
We then use these findings to propose a new construction for ``MFKDF2,'' a next-generation multi-factor key derivation function that can be proven to be end-to-end secure using the ESTMF. 
Finally, we discuss how MFKDF2 can be extended to support more authentication factors and usability features than the previous MFKDF construction, and derive several generalizable best-practices for the construction of new KDFs in the future.

\end{abstract}

\section{Introduction}
\label{sec:introduction}

Since the publication of ``Why Johnny Can't Encrypt'' in USENIX Sec '99 \cite{10.5555/1251421.1251435}, user-friendly client-side key management has remained one of the most difficult and widespread unsolved problems in the field of usable security.
Absent trusted hardware-based mechanisms for key management, many widely-used operating systems \cite{ios, windows}, software applications \cite{lastpass_arch, dashlane_arch}, and networking protocols \cite{wpa, wpa2} still turn to password-based key derivation functions (PBKDFs) like PBKDF2 \cite{rfc2898} or Argon2 \cite{argon2} for key management, limiting the security of derived keys to that of the passwords they are based on.
To address this deficiency, Nair and Song (USENIX Sec '23) \cite{NairSong2023} introduce the concept of a
Multi-Factor Key Derivation Function (MFKDF), a client-side KDF that incorporates other popular authentication factors like TOTP,
HOTP, and hardware tokens in the key derivation process.

While proven secure in the static, single-invocation threat model traditionally used to evaluate PBKDFs, MFKDF was shown by Scarlata, Backendal, and Haller (USENIX Sec '24) \cite{Scarlata2024} to leak entropy over multiple invocations.
At the core of this discrepancy is the fact that unlike stateless KDFs primitives (e.g., PBKDFs), MFKDF was the first \textit{inherently stateful} KDF, introducing new classes of multi-invocation vulnerabilities that made existing proof techniques inadequate.

This paper bridges the gap between static and dynamic security for stateful KDFs by introducing the Entropy State Transition Modeling Framework (ESTMF), a novel cryptanalytic technique for analyzing the dynamic, multi-invocation security of stateful cryptographic schemes. We show that ESTMF can be used to identify all known deficiencies within the original Nair and Song MFKDF construction (hereinafter NS23), and use these insights to propose MFKDF2, a new construction that is provably secure within the ESTMF model. MFKDF2 retains all of the features of NS23---including policy-based enforcement and factor recovery---while remedying its vulnerabilities and extending its capabilities with support for modern authentication factors like passkeys.

In this paper, we make the following contributions:
\begin{enumerate}
    \item We introduce the Entropy State Transition Modeling Framework (ESTMF), a new technique for analyzing the dynamic, multi-invocation security of stateful KDFs by finding all possible channels for entropy leakage (\S\ref{sec:estmf}).
    
    \item Using the ESTMF, we formally categorize all previously known attacks on NS23, and identify a new ``factor fungibility'' vulnerability related to factor ordering (\S\ref{sec:analysis}).
    
    \item We propose MFKDF2, a new MFKDF construction that achieves provable security even within the ESTMF model (\S\ref{sec:mfkdf2}). MFKDF2 extends NS23's MFKDF with support for new authentication factors like passkeys (\S\ref{sec:factors}) and new usability features (\S\ref{sec:features}), and can also be used for multi-factor credential hashing and generation (\S\ref{sec:modes}).

    \item We use the ESTMF to derive generalizable best practices for the secure construction and use of KDFs (\S\ref{sec:discussion}).
    
    
    
\end{enumerate}

\section{Background \& Related Work}
\label{sec:background}

Modern security architectures for applications like password managers \cite{1password_arch, lastpass_arch} and disk encryption software \cite{windows, ios} rely on password-based key derivation functions (PBKDFs) to solve the classical problem of usable key management \cite{10.5555/1251421.1251435} by deterministically transforming user passwords into cryptographic keys to provide end-to-end encryption.
Although PBKDFs can be made ``memory hard'' \cite{argon2, cryptoeprint:2016/027} to resist brute-force attacks, systems relying on this technique are still exposed to the general weaknesses of passwords as a single authentication factor, including their relatively low complexity \cite{florencio_large_2006}, frequent reuse \cite{password_reuse}, and vulnerability to password spraying and credential stuffing attacks \cite{credential_stuffing}.
While adding secondary authentication factors like HMAC-based One-Time Password (HOTP) \cite{rfc4226}, Time-based One-Time Password (TOTP) \cite{rfc6238}, or Out-Of-Band Authentication (OOBA) \cite{ooba} to such applications is now standard practice, doing so does nothing to further protect the confidentiality of user data in the event of a server compromise if password-based keys are still used as the underlying authentication or encryption mechanism.

\subsection{Multi-Factor Key Derivation}
\label{subsec:multi-factor_key_derivation}
The original Multi-Factor Key Derivation Function (MFKDF), as described by Nair and Song at USENIX Sec '23 (hereinafter NS23) \cite{NairSong2023}, was designed to solve this problem by incorporating all of a user's existing authentication factors (without modification), including HOTP, TOTP, and OOBA factors, directly into the client-side key derivation process. The construction was built on three core concepts:

\begin{enumerate}
    \item \textbf{Public Parameters:} A public state, initialized by a ``setup'' function, intended to be stored on an untrusted server but used for client-side key derivation. While similar to a ``salt'' in conventional PBKDFs, this value is also used to support dynamic factors where a new ``witness'' (e.g., a new TOTP code) is used for each derivation.

    \item \textbf{Factor Constructions:} A set of algorithms for converting witnesses from various authentication methods (including dynamic factors like TOTP) into static secret material suitable for deterministic key derivation.
    
    \item \textbf{MFKDF Construction:} The core MFKDF algorithm that invokes factor constructions to convert factor witnesses into static key material, combines this material, then invokes a memory-hard KDF to derive a static key, updating the public parameters in the process.
\end{enumerate}

NS23 then extends this basic approach to provide a ``threshold MFKDF'' construction, in which any $k$ of $n$ established authentication factors can be used to derive a key, allowing for client-side key recovery and reconstitution if some factors are lost, as well as a ``policy-based MFKDF'' construction that precisely enforces allowable combinations of factors.

In addition to providing candidate algorithms for standard, threshold, and policy-based MFKDF, and various factor constructions, NS23 established (and formally defined) three key security properties, which can be informally summarized as:

\begin{enumerate}
    \item \textbf{Correctness:} Providing a valid set of factor witnesses always derives the same ``correct'' key.

    \item \textbf{Safety:} Providing an invalid set of factor witnesses will almost certainly not derive the correct key.
    
    \item \textbf{Entropy:} Attackers should not be able to brute-force attack individual factors, instead having to search the entire space of all possible factor witness combinations, thereby providing an expontential security improvement.
\end{enumerate}

NS23 provides compelling proof that its suggested candidates satisfy their claimed security properties in the same threat model traditionally used to evaluate PBKDFs (i.e., across a single invocation). Importantly, they do not demonstrate the security of their scheme across multiple invocations.

\subsection{Cryptanalysis of MFKDF}
\label{subsec:cryptanalysis}
This static single-invocation security model used by NS23 to validate the properties of MFKDF, while adequate for standard PBKDFs, does not sufficiently account for an adversary who can observe and manipulate the KDF's public parameters over multiple invocations. As such, shortly after the publication of NS23, Scarlata, Backendal, and Haller published a red-team cryptanalysis of MFKDF in USENIX Sec '24 (hereinafter SBH24) \cite{Scarlata2024} which took advantage of these dynamic public parameters to several identify vulnerabilities in the construction that manfiest across multiple invocations.
The vulnerabilities identified in SBH24 are not disparate implementation bugs; rather, they are distinct classes of entropy leaks that cause the protocol to fail its stated security goals in a dynamic setting. We summarize the concerns below.

\subsubsection{Cryptographic Primitive Misuse}
This class of vulnerabilities stems from using cryptographic primitives in ways that are secure for one-time use but fail under repeated exposure or when plaintext is known.

\begin{itemize}
    \item \textbf{Key Recovery from Known Plaintext:} The ``HOTP Compromise Attack'' demonstrated that an attacker who compromises a factor's secret can combine it with the public ciphertext to recover the MFKDF master key.
    \item \textbf{Two-Time Pad Vulnerability:} The ``Share Recovery Attack'' showed that re-encrypting a static share with a new pad creates a classic two-time pad, allowing an attacker to recover the pads for low-entropy factors.
\end{itemize}

\subsubsection{State Management Failures}
This class of vulnerabilities arises from the protocol's failure to protect the integrity and forward secrecy of its own public state. An adversary who can manipulate the state can trick an honest client into making insecure decisions.
\begin{itemize}
    \item \textbf{Integrity Failure:} The ``Share Dilution Attack'' showed that an attacker can modify the policy in the public state, causing a user to weaken their own security.
    \item \textbf{Forward Secrecy Failure:} The ``Share Recovery Attack'' revealed that failing to regenerate shares during a \texttt{Recover} operation allows an old, compromised factor to remain a valid pathway for recovering a share.
\end{itemize}

\subsubsection{Side-Channel and Bias Leakage}
This class represents subtle form of entropy leakage, where information is extracted from biases or non-uniformities in the underlying cryptographic components. The ``Share Format Attack'' is a prime example, where the use of Shamir's Secret Sharing over a large prime field results in a non-uniform byte representation of the shares. This bias creates a perfect distinguisher, allowing an attacker to test factor guesses in isolation. A more subtle example is the ``HOTP Bias Attack,'' showing how a small statistical bias can be amplified over many observations to recover a factor's source key material.

\subsection{MFKDF Derivatives \& Impact}
\label{sec:impact}

Despite ultimately not living up to its stated security goals, the NS23 paper was initially well received by the academic community, being described as ``a significant advancement in cryptographic techniques... enabling the cryptographic enforcement of highly specific key generation policies'' \cite{HOANG2025104179} and winning the USENIX Sec '23 Distinguished Artifact Award \cite{usenixUSENIXSecurity}.
In a subsequent work, Nair and Song also showed that a well-formed MFKDF construction can form the basis of a secure ``Multi-Factor Credential Hashing Function'' (MFCHF) \cite{10190544}, allowing servers to simultaneously verify multiple standard user authentication factors while making server-side credential hashes $10^6$ to $10^{48}$ times harder for adversaries to crack without added latency or usability impact. They further showed that the same primitive can be used to construct a ``Multi-Factor Deterministic Password Generator`` (MFDPG) \cite{nair2023mfdpgmultifactorauthenticatedpassword}, enabling the creation of password managers with zero cloud or local storage, and essentially allowing users to upgrade password-only websites to MFA on the client side.

Beyond these applications, MFKDF has also been studied in the context of building more secure and user-friendly cryptocurrency wallets \cite{10174998, 10805375} that look and feel like centralized custodial wallets while in fact being trustless and decentralized. Upon evaluation, a large-scale user study published in ACM CHI '24 concluded that ``MFKDF-based applications outperform conventional key management approaches on both subjective and objective metrics, with a 37\% higher average SUS score (p < 0.0001) and 71\% faster task completion times (p < 0.0001) for the MFKDF-based wallet'' \cite{10.1145/3613904.3642464}.

Ultimately, the related publications demonstrate that a truly secure multi-factor key derivation construction would significantly improve the usable security of systems ranging from password managers to cryptocurrency wallets and uses cases ranging from client-side key derivation to server-side credential verification, motivating our work in this paper to design an ``MFKDF2'' that actually accomplishes the goals of NS23.

\subsection{Related Work}
\label{sec:related}

While NS23 remains the only serious attempt at constructing a client-side key derivation function with support for commonly used authentication factors like HOTP and TOTP, usable key management remains a highly-studied field \cite{hakimsecure, ide2025personhood} with many competing proposals \cite{10724844, cryptoeprint:2023/1785, 10.1007/978-3-031-91101-9_14}. Among these are proposals leveraging password-based key derivation \cite{argon2, cryptoeprint:2016/027}, biometric-based key derivation \cite{biokey, uzun_cryptographic_2021, seo_construction_2018, Soutar1998BiometricEU}, authenticated key exchange \cite{pointcheval_multi-factor_2008, liu_multi-factor_2011, chen_modular_2014, 6920371}, or secret sharing a key amongst several user devices or \cite{dalskov_2fe_2020} a trusted committee \cite{cryptoeprint:2016/144, cryptoeprint:2021/1522}.

Many key management solutions include trusted hardware or trusted committee assumptions, or focus on establishing a secure communication channel between two parties rather than deriving a deterministic encryption key.
Therefore, the design of an improved ``MFKDF2'' algorithm with the same setting and security goals of NS23 is essential to guarantee a drop-in replacement for PBKDFs and thereby achieve the myraid of applications and benefits described in \S\ref{sec:impact}.
\section{Entropy State Transition Modeling}
\label{sec:estmf}

The fundamental problem this paper addresses is the lack of a security model that can formally reason about a stateful KDF as a dynamic system; the vulnerabilities identified by SBH24 are all symptoms of this underlying issue in that they represent channels through which sensitive bits of entropy can either be revealed to an adversary over time, or be maliciously influenced by an adversary over time.

To truly reason about the security properties of a stateful KDF,
we must therefore analyze the system as a state machine and consider all possible channels where entropy leaks may occur. This includes not only the flow of information from the derived secret to the public state, but also the flow between input factors, the potential for an input's own public state to leak its secret, and the ability of an adversary to inject malicious data into the protocol's logic. This section introduces the
Entropy State Transition Modeling Framework (ESTMF), 
which provides the lens to model these flows
and ultimately provide a formal solution to this problem and prove that the MFKDF2 construction successfully seals them.



The ESTMF is built upon four principles that must be satisfied to prevent entropy leaks. We first define the general capabilities of an adversary interacting with any stateful KDF:

\begin{definition}
\label{def:adversarial_game}
Let an adversary $\mathcal{A}$ interact with a challenger $\mathcal{C}$ who manages a stateful cryptographic scheme. The \emph{adversarial interaction game} $\mathcal{G}$ proceeds as follows:
\begin{enumerate}
    \item Setup: The challenger initializes the scheme, resulting in an initial public state $\mathcal{B}_0$.
    \item Interaction: The adversary $\mathcal{A}$ receives $\mathcal{B}$ and may adaptively issue a polynomial number of queries $q$ corresponding to the scheme's public functions (e.g., \texttt{Derive}, \texttt{Recover}) with read and write access to the public state $\mathcal{B}$. For each query, the challenger computes the resulting state $\mathcal{B}'$ and returns it to the adversary providing a stream of public states $(\mathcal{B}_0, \mathcal{B}_1, \dots, \mathcal{B}_q)$.
\end{enumerate}
\end{definition}

This active adversary was a key missing component in NS23. While the paper argued that MFKDF state could be stored publicly, it is not clear to see this without considering an adversary who can observe and manipulate public state.  

\subsection{Master Secret Indistinguishability}
The highest level property within ESTMF is to ensure that there is zero entropy flows from the derivation process and master secret $M$ to the entire public state stream. A scheme that satisfies this is one whose public footprint, no matter how it evolves, reveals nothing about the core secret it protects.

\begin{definition}
\label{def:master_secret_indistinguishability}
A stateful KDF scheme has the property of \emph{Master Secret Indistinguishability (MSI)} if for any PPT adversary $\mathcal{A}$, its advantage in the following game is negligible. The \emph{MSI Game} proceeds as follows:
\begin{enumerate}
    \item Setup: The challenger generates two independent, uniformly random master secrets, $M_A, M_B \in \{0,1\}^\lambda$.
    \item Challenge: The challenger flips a random bit $b \in \{0, 1\}$. If $b=0$, the MFKDF instance will use $M_A$ and if $b=1$, it will use $M_B$.
    \item Interaction: The adversary plays the Adversarial Interaction Game $\mathcal{G}$ with the challenger, where the scheme instance is initialized using the previously master secret.
    \item Guess: Finally, the adversary $\mathcal{A}$ outputs a bit $b'$. The adversary wins if $b' = b$.
\end{enumerate}
The advantage of the adversary is defined as $\text{Adv}_{\mathcal{A}}^{\text{MSI}}(\lambda) = |\Pr[b' = b] - 1/2|$.
\end{definition}

For intuition, consider non-example of a secure block cipher $E$ operating in CTR mode on arbitrary plaintext but with respect to the MSI threat model. Here, an adversary may read the current counter state $j$, and force the challenger to take $j$ as the next counter, thereby reducing the security of $E$ to ECB mode. Hence, $E$-CTR is not an MSI secure scheme.

A catastrophic failure to meet this property was seen in the ``HOTP Compromise Attack,'' where a flaw in a single factor's public state component allowed an attacker with a compromised HOTP key to recover the final MFKDF key, and by extension, the master secret.

Crucially, MSI is not an independent property but rather the logical culmination of the other principles. Failures for a scheme to satisfy subsequent Defs. \ref{def:factor_independence_game}, \ref{def:leak_free_factor}, and \ref{def:state_integrity} creates an attack vector that an adversary can use to win the MSI game. Therefore, by positing MSI as the required security goal for a stateful KDF, the ESTMF forces a designer to systematically identify and seal all potential channels for entropy leakage.

\subsection{Entropy Flow Between Factors}
\label{subsec:entropy_flow_between_factors}
Another potential leak is between factors themselves which may be mediated by an adversary. A secure scheme must ensure that compromising one factor does not open a side-channel for entropy to flow from other, uncompromised factors, effectively preventing collateral damage. Clearly, if one compromised factor can degrade the security of another, we lose the desired entropy property given in \S\ref{subsec:multi-factor_key_derivation}.

\begin{definition}
\label{def:factor_independence_game}
A KDF is Factor-Indistinguishability under Chosen Message Attack (Factor-IND-CMA) secure if for any probabilistic polynomial-time (PPT) adversary $\mathcal{A}$, its advantage in the following game is negligible. 
\begin{enumerate}
    \item Setup: The challenger creates a full MFKDF instance with a set of factors $\mathcal{F}$, which includes two distinct "challenge" factors, $F_A$ and $F_B$, of the same type. The adversary $\mathcal{A}$ chooses a subset of factors $\mathcal{F}_{\text{compromised}} \subset \mathcal{F} \setminus \{F_A, F_B\}$ and is given their complete secret material.
    \item Challenge: The challenger flips a random bit $b \in \{0, 1\}$. If $b=0$, the MFKDF instance will use $F_A$ as the active challenge factor. If $b=1$, it will use $F_B$.
    \item Interaction: The adversary plays the role of a malicious server and interacts with the challenger, who simulates an honest client, as defined in Game $\mathcal{G}$.
    \item Guess: Finally, the adversary $\mathcal{A}$ outputs a bit $b'$. The adversary wins if $b' = b$.
\end{enumerate}
The advantage of the adversary is defined as $\text{Adv}_{\mathcal{A}}^{\text{Factor-IND-CMA}}(\lambda) = |\Pr[b' = b] - 1/2|$.
\end{definition}

The ``Share Recovery Attack'' (two-time pad aspect) violated this principle by creating a dependency between states of a single factor, while the ``Factor Fungibility'' attack violated it by creating a dependency between two factors.

\subsection{Entropy Flow from Input to Public State}
Another leaky instance is the models the potential entropy flow from a factor's secret $\sigma$ to its own public state $\beta$ of an individual factor. This localized principle is therefore required of the local factor construction portion of MFKDF as opposed to the global key feedback mechanism. Note that this property implies that any factor construction, no matter the factor type, cannot betray the key used in its state transitions. Hence, this is what allows us to deem a specific factor construction to be free of entropy leaks. 


\begin{definition}
\label{def:leak_free_factor}
A factor construction is \emph{Factor-Key Indistinguishability (KI) secure} if for any PPT adversary $\mathcal{A}$, its advantage in the following game is negligible. 
\begin{enumerate}
    \item Setup: The challenger generates two independent, uniformly random keys, $K_0, K_1 \in \{0,1\}^\lambda$.
    \item Challenge: The challenger flips a random bit $b \in \{0, 1\}$.
    \item Interaction: The adversary plays the Adversarial Interaction Game $\mathcal{G}$ with the challenger, where all state transitions for a single factor are computed using key $K_b$.
    \item Guess: Finally, the adversary $\mathcal{A}$ outputs a bit $b'$. The adversary wins if $b' = b$.
\end{enumerate}
The advantage of the adversary is defined as $\text{Adv}_{\mathcal{A}}^{\text{Factor-KI}}(\lambda) = |\Pr[b' = b] - 1/2|$.
\end{definition}

The ``Dynamic Factor Attack'' is a prime example of this failure, where a factor's public state $\beta$, combined with a single witness $W$, directly leaks the factor's source key material $\kappa_F$. The ``HOTP Bias Attack'' is a subtle example where the public state leaks information over many observations.

\subsection{Adversarial Flow into the Protocol}
Due to the public nature of the MFKDF state, we must also address the case where an adversary can inject their own information into the protocol's logic. It is quite clear that if an adversary could gain complete control over a client's code execution, then the derived key can immediately be compromised. Therefore, we restrict to the case where an adversary only has access to public state, but not the protocol that the client is running. 

\begin{definition}
\label{def:state_integrity}
A stateful scheme provides \emph{state integrity} if for any PPT adversary who can modify the public state, the probability that an honest client accepts a tampered state as authentic is negligible.
\end{definition}

The ``Share Dilution Attack'' and ``Parameter Tampering'' attacks are canonical examples of this, where an attacker injects a modified policy or KDF parameters to trick an honest client into weakening its own security.

\section{ESTMF Analysis of NS23}
\label{sec:analysis}
For the following, we suppose that an adversary wishes to degrade the security of NS23 by any means necessary. We let the adversary follow that of Def. \ref{def:adversarial_game} so that they may have access to a public state stream $\{\mathcal{B}_0,\mathcal{B}_1,\dots,\mathcal{B}_q\}$ of their own choosing where each $\mathcal{B}_j$ has a collection of factors that are $\sum_i m_{i,j}$-entropy. The adversary may also suggest the exact factors used to run the MFKDF derivation. We also suppose the challenger (the user deriving an MFKDF key) is ``rational'' and will reject running \texttt{Recover} to a new state (or policy) that degrades their security if they can detect it.

\subsection{State Integrity}
Consider the implications of a stateful cryptographic scheme lacking state integrity (Def. \ref{def:state_integrity}). Immediately we find a vulnerability in the adversarial game. If the adversary intercepts $\mathcal{B}_j$ and produces a malicious $\mathcal{B}_{j+1}$, they can call upon the challenger to run \texttt{Recover} with a chosen subset of factors $\mathcal{F}_C$. This leaves the challenger to input their correct witnesses $\mathcal{W}_C$.
Specifically, an adversary could insert $\mathcal{B}_{j+1}$ such that:
\[
    \sum_i m_{i,j+1} < \sum_i m_{i,j},
\]
so long as $\mathcal{W}_C$ are still valid, thereby decreasing the entropy of the next key derivation. One such way to do this would be to remove unused factors in the derivation process all together or, in a tougher to distinguish case, reduce the threshold required to derive the key.
This entropy leak exists at the protocol level even for a rational challenger as detection is not guaranteed. In practice, by trusting a server holding a copy of the MFKDF state, a rational client may simply not detect that the state was unexpectedly modified. To mitigate this vulnerability, we must adopt a change that forces the client to check for self-certified integrity prior to updating to a new state $\mathcal{B}_{j+1}$. To do so, we propose signing the state with the derived key $K$ via a secure MAC such as HMAC-SHA256. 
Afterward, if the attacker mounts a similar attack, regardless of the suggested state $\mathcal{B}_{j+1}$, the challenger detects a fraudulent state via the MAC and will halt the \texttt{Recover} process before a new state is produced, as this requires \texttt{Derive} to run first.

\subsection{Factor Flow}
Suppose that our adversary $\mathcal{A}$ collects a set of compromised factors $\mathcal{F}_C$ associated to state $\mathcal{B}_{j}$; does this reduce the entropy of the remaining factors? A factor construction satsifying Factor-IND-CMA (Def. \ref{def:factor_independence_game}) seals information from flowing from itself over a state change and, likewise, from one factor to another even within the same state or subsequent state. 

In a simple case, perhaps the adversary holds valid factor witnesses $\mathcal{W}_C$ but is unaware of the precise input locations of each witness. If the NS23 implementation provided the entropy property (\S\ref{subsec:multi-factor_key_derivation}), then the attacker must brute-force and exhaustively try all possible combinations of witnesses. 

This is indeed not satisfied by the NS23 implementation. For $n$ of $n$ threshold MFKDF or policy-based MFKDF with key stacking, we derive a shared secret $S$ via XOR based sharing of shares $s_i$:
\[
    S = \bigoplus_i s_i.
\]
Then since shares $s_i$ are derived by $s_i = \operatorname{pad}_i \oplus \operatorname{KDF}(\sigma_i)$, we find that factors leak information between each other in this construction as the XOR operation above is commutative so the order in which the $\operatorname{pad}_i$ are applied is irrelevant. Hence, the ESTMF framework determines a combinatorial reduction in entropy due to the fact that factor submission order is unspecified. This suggests Shamir's Secret Sharing (SSS) must be used even in the $n$ of $n$ case.

Furthermore in the threshold case, we can see that by having even a single compromised factor $F_C$, the attacker knows all about the encrypted share associated to that factor such as the encrypted share $c_C$ as well as the secret $s_C$ and encryption pad $\operatorname{pad}_C$ via $c_C = s_C \oplus \operatorname{pad}_C$. At a later time, it's clear to see that information from this compromised factor remains unless the share itself is updated. 
Even if the factor is not compromised, the entropy leak happens upon state update. If the adversary observes the public state, it notices encrypted shares $c_{i}$ and $c_{i+1}$ and finds
\[
c_i \oplus c_{i+1} = (s_i \oplus \operatorname{pad}_i) \oplus (s_i \oplus \operatorname{pad}_{i+1}) = \operatorname{pad}_i \oplus \operatorname{pad}_{i+1}.
\]
which implies that this single factor is no longer "independent" in the sense that it can now be attacked on its own right. Hence, a state update allows a factor to leak entropy "back to itself" which reduces entropy of the whole scheme.

\subsection{Factor Secret Flow}
This channel models the flow of entropy from a factor's internal secrets ($\sigma_i$ or $\kappa_{i}$) to an adversary. The leak arises when the design of the MFKDF state machine creates a new attack surface that undermines the security of an individual factor construction, even if that factor is considered secure in other contexts. The ESTMF principle of Factor Key-Indistinguishability (Factor-KI) (Def. \ref{def:leak_free_factor}) is the property required to seal this channel. A factor construction that does not satisfy this game can expose its secrets in two ways.

First, a protocol may create a direct leak by storing public helper data that, when combined with a single valid witness, reveals the factor's underlying secret. This is a significant risk for dynamic factors, where the construction must convert a changing witness into a static secret. If the public state contains an offset or value that is algebraically related to both the witness and the secret, then a single observation of a valid witness can lead to a total compromise of that factor's entropy. Similarly, due to an algebraic relation such as a stored offset, with one valid witness, other stored witnesses could be found. 

Second, a protocol can create a statistical leak over time. Even if a single public state reveals no information, the stream of public states across many invocations may expose subtle biases or non-uniformities in a factor's construction. An adversary who can observe this stream can perform statistical analysis to gradually learn information about the factor's secret. ESTMF's focus on analyzing the entire state stream $(\mathcal{B}_0, \mathcal{B}_1, \dots, \mathcal{B}_q)$ forces designers to consider these long-term, low-bandwidth leaks that static analyses would miss.

\subsection{Master Secret Flow}
The principle of Master Secret Indistinguishability (Def. \ref{def:master_secret_indistinguishability}) is the essential security property of the ESTMF. It models the total entropy flow from the master secret $M$ to an adversary who can observe and modify the entire public state stream over the protocol's lifetime. A scheme that is MSI-secure is one whose public footprint is computationally independent of the core secret it protects, even after many state transitions.

MSI bridges each of the principles and sets the highest bar for defense. For example, a failure of state integrity allows an attacker to inject a tampered state that may cause the protocol to behave differently depending on the master secret, creating a distinguisher that breaks MSI indirectly. Similarly, a failure of Factor-KI could allow a compromised factor to leak information about a share, which is derived from the master secret. Lastly, Factor-IND-CMA can catch leaks between the different factors themselves, such as lack of requiring order, which again causes a failure to be MSI secure.

Yet there are attacks that fall outside of the scope of the previous definitions that MSI would detect as well. Given that the adversary may have access to specific factor secrets $\sigma$, this must not also leak any further information about the MFKDF derived key $K$ than knowledge of the associated witness $W$ to the compromised factor. In a sense, MSI also detects ``privilege escalation'' leaks whereby an adversary bypasses other factors by virtue of having an subset that should not be able to derive $K$.
Specifically, we note that NS23 uses the derived key $K$ to encrypt factor secrets by
$c = \sigma \oplus K$
which means an attacker that knows $\sigma$ (by compromise) and $c$ (because it is public) then knows $K$ by the invertibility of $\oplus$. It is clear that we must seal this channel by requiring an encryption that has an algebraic obfuscation property such as a Pseudo-Random Permutation (PRP) as opposed to a commutative XOR.
\section{MFKDF2}
\label{sec:mfkdf2}
The cryptanalysis of the original MFKDF using ESTMF revealed several classes of vulnerabilities which are also outlined in \S\ref{subsec:cryptanalysis}. The MFKDF2 construction is designed to systematically satisfy ESTMF and therefore solve each vulnerability described by SBH24. This section details the specific security improvements in MFKDF2 and explains how they mitigate the known attacks.

\subsection{Hardened Cryptographic Defaults}
To prevent parameter tampering attacks, MFKDF2 removes security-critical parameters from the public state. The KDF is fixed to Balloon Hash \cite{CorriganGibbs2016} using SHA3-256 \cite{FIPS202}, and the derived key length is fixed to 256 bits. This ensures that an attacker with control over the public state cannot trick the client into using a weaker set of cryptographic primitives. In particular, we avoid the XOF attack which allows for an attacker to derive a key a few bits at a time.

\subsection{Per-Factor Salting}

MFKDF2 introduces mandatory \emph{per-factor salting}. Each factor instance is combined with a unique, randomly generated salt that is stored in the public state before its secret material is used. This makes each factor instance cryptographically unique, preventing any reordering or swapping attacks. This property can be formally proven by showing that the output of the salted construction is computationally independent of the order of the factors.

\subsection{Information Theoretic Secret Sharing}

MFKDF2 explicitly implements bytewise secret sharing over the finite field $\operatorname{GF}(256)$. We do formally prove in Lem. \ref{lem:sss_indistinguishability} that sharing over $\operatorname{GF}(2^q)$ suffices for a key of 256 bits if $2^q$ divides 256. This ensures that the SSS is information-theoretically secure and that the shares are uniformly random bytes, preventing any potential biases or side-channels. A practical consequence of operating over $\operatorname{GF}(256)$ is that a maximum of 255 unique shares can be generated for any given secret byte (see \cite[Theorem 22.3]{BonehShoup2020}).

\subsection{Encryption with a Secure PRP}
The entire class of vulnerabilities related to cryptographic primitive misuse is solved by replacing the insecure one-time pad (XOR) with a secure Pseudorandom Permutation (PRP). Results for which follow from a standard result (Thm. \ref{thm:prp_kpa_security}) and lead us to Thm. \ref{thm:prp_share_encryption}. Hence, in MFKDF2 we specify the use of AES-256 as a block cipher to encrypt shares.

\subsection{State Integrity via Self-Referential MAC}
To prevent all forms of state tampering, MFKDF2 mandates that the entire public state be authenticated with a MAC (specifically HMAC-SHA256) tag $\mathcal{T}$ computed with the derived key. As proven in Prop. \ref{prop:state_integrity}, this self-referential check creates an unbreakable loop. An honest client will only accept a state if the key it derives from that state correctly validates the state's integrity. Any unauthorized modification by an attacker will cause this check to fail, forcing the client to abort. We find that Prop. \ref{prop:state_integrity} provides this security in the case of a rational client running the MFKDF2 algorithm.

\subsection{Timing Oracles for Dynamic Factors}
\label{subsec:timing_oracle}
There is a fundamental concern when using dynamic factors to derive a static key. In essence, dynamic material is converted to static material so that a key may be derived entirely by a client. This creates an issue whereby an attacker holding a copy of an old policy having seen one valid TOTP code (e.g., by shoulder-surfing) can write the down time this code was valid, and then access this factor share at any point in the future. This does not mean that the exponential security guarantee of NS23 is broken, but it does mean that the typical threat model assumed by TOTP is not satisfied by NS23, which is certainly a point of confusion for an end user, as the expectation is that a single TOTP witness is ephemeral. In NS23, a TOTP code effectively serves as an additional 20-bit (static) password, whereas in the typical model, a code is valid only for the window since a trusted server is able to verify the time alignment in its half of the protocol. 

In MFKDF2, we mitigate this concern
by introducing an optional timing oracle, a stateless, online service that assists in the TOTP factor derivation process, enforcing time alignment and maintaining exponential security guarantees solely by introducing the risk of a liveness failure for a single factor.

Under the hood, the timing oracle essentially stores an internal secret ``pepper'' value, and runs its own copy of the TOTP protocol using an HMAC of the pepper value and a user-provided key. Upon setup, the TOTP value for each timestep is subtracted from the TOTP value provided by the user's TOTP application, simultaneously perfectly offsetting the attackable HOTP biases described by SBH24
and ensuring that expired TOTP codes can no longer be used to derive a key unless the oracle was queried at that instant.


\subsection{Security Under ESTMF}
If we take MFKDF2 as a protocol adapting each of the previous security improvements given in this section, we show it is a provably secure system under the ESTMF framework. Let us now walk through the arguments.

First, we note that by virtue of fixing the key size to 256 bits and fixing the underlying KDF to Balloon-SHA3-256, the addition of the self referential MAC prevents parameter tampering via Prop. \ref{prop:state_integrity}. This greatly reduces the adversary's capabilities in concert with subsequent security guarantees.

Next, by removing the use of a one-time-pad for encryption, the replacement with a PRP (i.e., a block cipher such as AES-256 suffices) prevents attacks that use algebraic properties of the one-time-pad to leak information about the underlying encryption key. This is a standard result given by Thm. \ref{thm:prp_kpa_security}, but it also implies Thm. \ref{thm:prp_share_encryption} which tells us that using a PRP for encryption of factor secrets or share encryption does not leak entropy through their own channel. 

Additionally, even without modifying share encryption, Thm. \ref{thm:share_updating} implies that simply updating shares after a recovery operation guarantees forward secrecy. Combined with the previous result, this gives us a security-in-depth approach.

Continuing, given that the collection of factors used for MFKDF2 satisfy the Factor-KI (Def. \ref{def:leak_free_factor}), we can show that the MFKDF2 scheme itself is secure in the $n$ of $n$ case via Thm. \ref{thm:factor_independence_n_of_n} since this does not require any secret sharing. 

For the threshold case, we show via Lemma \ref{lem:sss_indistinguishability} that for $b$-bit keys, SSS over $q$-bit chunks where $b$ is a multiple of $q$ is provides indistinguishable shares that are also information theoretically secure, so long as less than $2^q$ shares are issued. For MFKF2, this implies SSS for a 256 bit key over $\operatorname{GF}(256)$ is secure. Following the previous theorem and using this lemma, we prove the threshold MFKDF2 is secure in Thm. \ref{thm:factor_independence_threshold} which implies the policy MFKDF2 is secure as well.

Given MFKDF2 and its associated factor constructions satisfy the above properties, we prove the final result of Thm. \ref{thm:mfkdf2_achieves_msi}. All together, this means that the new construction of MFKDF2 itself is secured under the ESTMF framework.

\section{New Factors in MFKDF2}
\label{sec:factors}

To broaden the use cases of MFKDF2, we also provide many new factors to use. These allow for modern authentication primitives such as Passkeys, OIDC, fuzzy or enclave-based factors for hardware or biometrics, and other factors typically used by mobile devices or by other proximity aware hardware.

\subsection{Passkeys (WebAuthn-PRF)}
Passkeys are a modern form of authentication that uses a signing key (typically from \texttt{secp256r1}). Usually, these passkeys are stored in a password manager \cite{1PasswordPasskeysProduct} and perhaps on a secure element \cite{AppleHardwareSecurity}. The brief overview is that a passkey is given a random challenge, it signs the random challenge and the server verifies the signature since it has a reference to the public key. 
Since passkeys are becoming a widespread, high-entropy authentication method, they are an excellent candidate for use as a factor in MFKDF2. A naive implementation using standard WebAuthn signatures is not possible, as the signatures are intentionally non-deterministic to prevent replay attacks, which makes it impossible to derive a static key from them.

Instead, we leverage the {WebAuthn PRF extension} (\texttt{prf}), which is specifically designed to enable the derivation of deterministic, secret values from a FIDO2 authenticator. The construction closely mirrors the HMAC-SHA1 challenge-response factor from the original MFKDF paper \cite{NairSong2023}. The goal is to use the \texttt{prf} extension to derive a static \texttt{prf\_key}, which then serves as the factor's secret material (see Alg. \ref{alg:mfkdf2_passkey_factor}).

\subsection{Fuzzy Encryption Factors}
To support factors with inherently noisy or variable inputs, such as biometrics, geolocation, or behavioral patterns, MFKDF2 introduces support for fuzzy encryption. This allows a stable, high-entropy secret to be derived from an input that is merely ``close enough'' to an original enrollment sample, rather than being identical.

Beyond biometrics, this technique enables several novel, noise-tolerant factors. For example, \emph{geolocation} can be used to ensure that key derivation is possible if the user's current GPS coordinates are within a predefined radius of an enrolled location. Second, \emph{behavioral biometrics} could ensure that a user's unique typing cadence (keystroke dynamics) or mouse movement patterns are captured to prevent bot login.

The core cryptographic primitive that enables this is a \emph{fuzzy extractor} \cite{DodisOstrovsky2008}. A fuzzy extractor consists of two algorithms, \texttt{Gen} and \texttt{Rep}, which operate as follows:
\begin{itemize}
    \item $\texttt{Gen}(w) \rightarrow (K, P)$: The generation algorithm takes a sample $w$ (e.g., an initial fingerprint scan) as input. It outputs a uniformly random cryptographic key $K$ and a public piece of helper data $P$. The sample $w$ can be discarded.
    \item $\texttt{Rep}(w', P) \rightarrow K$: The reproduction algorithm takes a fresh, noisy sample $w'$ (e.g., a new fingerprint scan) and the public helper data $P$. If the distance between $w$ and $w'$ is within a certain error tolerance $t$ (i.e., $\text{dist}(w, w') \le t$), it reliably reconstructs the exact same key $K$. 
\end{itemize}
Using the above, we define the relevant \texttt{FactorSetup} and \texttt{FactorDerive} algorithms in Alg. \ref{alg:mfkdf2_fuzzy_factor}.

\subsection{Risk-based and Environmental Factors}
\label{subsec:risk-based_environmental_factors}
Many types of identifying information related to the user's environment can be used to improve the user experience of MFKDF2 while also increasing security. Some of these factors, like device profiles, can be noisy and may be best handled with fuzzy extractors. Yet, in the case of strict requirements, these factors are static and follow the same algorithm as a password or security question, but will not require users to engage directly. 
The typical use case is adaptive authentication. In the case that a collection of low-risk factors are present (e.g., a known device profile from a familiar location), the policy can allow the user to bypass a stronger secondary factor like a TOTP. Alongside this example, a user may wish to add a new device to their list of risk-based factors which they can do by using the stronger secondary factors to allow for \texttt{Recover} to update the policy.

Some possible factors analyze the origin and reputation of the user's network connection. For example, an IP address can be checked to ensure it is on a known list of trusted IPs for a given user. Second, device fingerprinting that uses a detailed profile of the user's hardware (CPU class, memory), software (OS version, installed fonts), and browser configuration (plugins, screen resolution, language settings) is also useful as a factor. One could also use a time-of-day window as well. It should be noted that the intent of such factors is to layer on security and/or improve user experience.

\subsection{Mobile Factors}
Mobile devices are ubiquitous and have become a cornerstone of modern multi-factor authentication, acting as a trusted device that users almost always have with them. MFKDF2 leverages this by introducing support for mobile-centric authentication flows that are both secure and user-friendly.

\subsubsection{QR Code Factors}
Secure QR Login (SQRL) \cite{grcGRCapossxA0xA0SQRLSecure} is a protocol for secure authentication using QR codes. We can adapt this challenge-response mechanism to create a strong MFKDF2 factor.

We adapt the challenge-response mechanism of SQRL to create a high-entropy, QR code-based factor. During setup, the user's mobile device registers its public key with the MFKDF2 policy. To derive the key, the client generates a unique, random challenge and displays it as a QR code. The user scans this code with their mobile device, which in turn signs the challenge with its private key. This deterministic signature serves as the static factor material for the derivation, in a flow analogous to the HMAC-based hardware token factor but using asymmetric cryptography (see Alg. \ref{alg:mfkdf2_sqrl_factor}).

\subsubsection{Push Notification Factors}
Push notifications provide a seamless way to use a mobile device as an "out-of-band" authenticator by turning a user's approval on a second device. During setup, the mobile device's public key is registered with the MFKDF2 policy. To derive the key, the client sends a unique, random challenge to the user's device via a push notification service. Upon user approval, the mobile app signs the challenge with its private key. This signature is then returned to the client through a secure channel and serves as the static secret factor material.

\subsection{Biometric Factors via Secure Enclaves}
Biometric authentication is a prime use case for the fuzzy encryption techniques discussed previously. However, for security and privacy, modern platforms do not expose raw biometric data. Instead, they provide high-level APIs that leverage on-device {Secure Enclaves} or {Trusted Execution Environments (TEEs)}. MFKDF2 integrates with these systems (e.g., Apple's Face/Touch/Optic ID \cite{AppleHardwareSecurity}, Windows Hello \cite{WindowsHello}) to create robust, hardware-backed factors.

MFKDF2 uses secure-hardware based factors in the following way: during enrollment, the client instructs the device's secure enclave to generate and store a new high-entropy secret, which is cryptographically bound to the user's biometric data and never leaves the enclave.
To use the factor, the client requests the secret from the enclave. The enclave, in turn, prompts the user for biometric authentication. If the scan is a successful match against the enrolled data, the enclave releases the secret to the client. This released secret then serves as the static, high-entropy factor material for the MFKDF2 derivation. This approach provides a seamless user experience while ensuring that sensitive biometric information is never processed or stored outside of the device's secure hardware.

\subsection{Proximity Factors}
Deriving keys based on physical proximity to a device or location is a powerful use case, particularly for physical access control. MFKDF2 can integrate with common short-range communication technologies like Near-Field Communication (NFC) and Radio-Frequency Identification (RFID) to enable these scenarios. While the physical properties of these technologies differ, they can be abstracted to the same secure, cryptographic challenge-response protocol.
The core mechanism prevents simple cache and replay attacks. During setup, a physical tag (e.g., an NFC sticker or an RFID fob) is provisioned with a high-entropy secret key. To derive, the MFKDF2 client (a reader) generates a fresh, random challenge which is stored in the public state for the next derivation. The client transmits this challenge to the tag, which computes a deterministic response, typically an HMAC of the challenge with its secret key. This response is transmitted back to the client and serves as the static, high-entropy factor material.
The primary difference between the two factors lies in their application. NFC, with its very short range of less than 4 cm, is ideal for intentional, ``tap-to-derive'' interactions, such as unlocking a personal workstation or a specific door. RFID operates at a longer range, making it well-suited for seamless, hands-free use cases, such as granting a vehicle access to a garage or tracking assets within a secure facility. 

\subsection{OAuth 2.0 and OpenID Connect Factors}
While OAuth 2.0 is an authorization framework, the {OpenID Connect (OIDC)} identity layer built on top of it provides the necessary primitives to create a secure MFKDF2 factor, allowing a user to leverage the high security of their major online accounts (e.g., Google, Apple). The construction uses the OIDC provider as a trusted third party to sign a challenge. The MFKDF2 client generates a large, random \texttt{nonce} and initiates an OIDC authentication flow, including the \texttt{nonce} in the request. After the user authenticates with the provider, the provider issues a signed JSON Web Token (\texttt{id\_token}) that contains the \texttt{nonce}. The signature of this token is a deterministic, high-entropy value that serves as the static factor material for the derivation.

\section{New Features in MFKDF2}
\label{sec:features}

Beyond addressing the known (and newly-discovered) vulnerabilities of its predecessor, MFKDF2 introduces two new features designed to enhance its long-term security, usability, and ease of adoption. Using a mechanism for key encapsulation, we can allow for increasing difficulty of the underlying KDF and with minimal entropy trade off we can provide a user hints on which factor may be correct.

Though not mentioned in the original manuscript, the MFKDF.js implementation supports an "envelope" API which allows users to securely store other secrets alongside an MFKDF derived key \cite{MFKDFDocsEnveloping}. This interface, in particular, allows for transferring a known key to an MFKDF encapsulated key. For example, it is common to use a BIP-39 \cite{BIP0039} seed phrase to derive private keys for cryptocurrency wallets. These seed phrases are often 12, 18, or 24 words that users store on various materials such as paper or metal and, since the value associated to these keys is quite high, they may often be duplicated or placed in safety deposit boxes. Instead, one could consider storing the seed phrase as an enveloped secret alongside an MFKDF derived key, or even just store the associated private key as an enveloped key. This allows a cryptocurrency user to keep their same private key, but remove the legacy storage media for a socially-recoverable MFKDF "vault". 

\subsection{Adaptive Security Parameters}
While MFKDF2 mandates strong, opinionated defaults to prevent parameter tampering, the above enveloping method provides a mechanism for cryptographic agility. This allows the security of the underlying KDF to be strengthened over time as computational power increases without changing the an enveloped key. 

The KDF parameters (e.g., Balloon hash memory and time costs) are stored in the public state and protected by the self-referential MAC. A user can propose an upgrade to stronger parameters. The protocol for this is as follows:
\begin{enumerate}
    \item The user derives their key $K$ using the current, authentic state with its existing KDF parameters. The MAC check ensures this state is valid.
    \item The client then creates a new public state, $B'$, which is identical to the old state but contains the new, stronger KDF parameters.
    \item The client computes a new tag $T' = \operatorname{MAC}(K, B')$ and sends the new state $(B', T')$ to the server.
\end{enumerate}
Because this operation requires the user to successfully derive the key $K$, only a legitimate user can authorize an increase in the KDF's computational cost. An attacker cannot weaken the parameters, as this would invalidate the MAC. 

\subsection{Probabilistic Factor Hints}

An additional feature of our MFKDF2 construction is the ability to store a small number of bits of entropy corresponding to one or more input factors in the MFKDF2 public policy, reducing the strength of the derived key by that number of bits, but allowing users to probabilistically validate.

For example, consider a 2-factor MFKDF setup consisting of a password and a TOTP code. Typically, if a user were to enter either factor incorrectly, they would derive the wrong key, but would not know which factor was incorrect. However, via MFKDF2's probabilistic hints feature, storing just 7 bits of the password's salted hash in plaintext would give legitimate users a greater than 99\% chance of knowing whether their password was the incorrect factor while offering a negligible amount of additional information to most adversaries.


\section{Modes of Operation}
\label{sec:modes}

As detailed in \S\ref{sec:impact}, there were two subsequent papers released that used the MFKDF construction to produce further primitives: MFCHF \cite{NairSong2023_mfchf} and MFDPG \cite{NairSong2023_mfdpg}. In MFKDF2, we formalize MFCHF2 and MFDPG2 as ``modes of operation'' to simplify the usage and analysis of the MFKDF2 ecosystem.

\subsection{Architecture}
The core MFKDF2 algorithm serves as a foundational primitive for deriving a high-entropy master key from a multi-factor policy. The different modes of operation leverage this master key for distinct cryptographic purposes. This modular architecture allows a single, provably secure key derivation process to be used as a building block for a variety of applications, from server-side credential hashing to client-side deterministic password and passkey generation. This simplifies the overall design and ensures that improvements to the core MFKDF2 protocol are inherited by all modes of operation.

\subsection{Multi-Factor Credential Hashing}

While MFKDF2 secures client-side secrets, \emph{Multi-Factor Credential Hashing Function (MFCHF2)} addresses the critical server-side problem of protecting stored credentials post-breach. MFCHF2 provides asymmetric resistance to brute-force attacks by incorporating the entropy of a user's secondary factors directly into the stored password hash. This significantly increases the cost for an offline attacker without increasing the latency for a legitimate user.

The original MFCHF construction \cite{NairSong2023_mfchf} uses NS23's factor constructions to convert a dynamic witness (like a TOTP code) into a static, secret `target` value. This `target` is then hashed alongside the user's password using an adaptive hash function like Argon2. MFCHF2 follows this same powerful paradigm but inherits the hardened security of the underlying MFKDF2 construction, including its use of secure PRPs and robust state management. This prevents the specific cryptographic leaks found in the original NS23 factor constructions.

Furthermore, MFCHF2 improves upon the original by leveraging the full policy engine of MFKDF2. Instead of storing disparate public parameters for each factor, the server can store the single, opaque, and integrity-protected MFKDF2 public state. This has two advantages: first, it simplifies the server-side logic, and second, it obscures the details of the user's multi-factor policy from an attacker who has breached the database. An attacker can no longer immediately tell which users have simple password-only protection versus those with complex, multi-factor policies, forcing them to treat every credential as if it were maximally protected.

\subsection{Deterministic Pass(word/key) Generation}

In addition to deriving a master key, the MFKDF2 paradigm can be used as a foundation for a \emph{Multi-Factor Deterministic Password Generator (MFDPG2)}. This addresses the significant security risks of traditional password managers by generating high-entropy, site-specific passwords on demand, rather than storing them in a vault that can be breached. MFDPG2 extends the original MFDPG concept \cite{NairSong2023_mfdpg} to support not only the generation of traditional passwords but also the derivation of cryptographic keys, such as those used for passkeys.

The original MFDPG algorithm uses the key derived from NS23 as a master secret. This secret is then combined with a domain identifier (e.g., a website's URL) and fed into another KDF to produce a deterministic, site-specific seed. This seed is then used to generate a password that complies with the site's specific requirements (e.g., length, character sets). MFDPG2 follows this same robust, multi-stage process but inherits the enhanced security of the underlying MFKDF2.

\subsubsection{MFKDF2 for Passkey Derivation}
Beyond using passkeys as an input factor, MFKDF2 can also serve as a mechanism for deriving the passkey itself. This presents a powerful new paradigm that mitigates the need for storing a passkey's private key in a password manager (where its security is tied to a single factor) or on a single hardware device (where it can be lost), and instead forces the private key to be deterministically re-derived from multiple factors on demand. This is a specific mode of operation for MFDPG2 that is similar to the password format requirements of web applications.
The primary requirement for deriving a passkey is that the value must be a valid scalar for the chosen elliptic curve (e.g., \texttt{secp256r1}). A raw 256-bit output from a KDF is not guaranteed to be within the valid range of the curve's order. To solve this, we use a deterministic rejection sampling method, as recommended by standards like FIPS 186-5 \cite{FIPS186-5}:

\begin{enumerate}
    \item \textbf{Initial Derivation:} The MFKDF2 key $K$ is derived as usual from the user's chosen factors. This key $K$ serves as a high-entropy seed.
    \item \textbf{Candidate Generation:} We use an HKDF-like construction, using $K$ as the input keying material (IKM) and a domain separation tag (e.g., "MFDPG2-Passkey-Derivation"). We can then generate a sequence of candidate scalars $d_i$ by using an incrementing counter as the ``info'' parameter:
    $d_i = \operatorname{HKDF-Expand}(K, \text{"MFDPG2-Passkey"}, i)$
    \item \textbf{Rejection Sampling:} For each candidate $d_i$, we check if it is in the valid range $[1, n-1]$, where $n$ is the order of the curve's base point. The first candidate that falls within this range is selected as the passkey's private key.
\end{enumerate}
This construction provides a robust, recoverable, and purely cryptographic method for managing passkeys, elevating their security to the full strength of the user's multi-factor policy.

\section{Discussion}
\label{sec:discussion}

The Secure Hash Algorithm 1 (SHA1) \cite{rfc3174}, first released in 1995, was once considered a cornerstone of digital security, advancing the field of cryptographic hashing beyond the vulnerable MD5 family. Nearly ten years later, the popularity of SHA1 had finally motivated sufficient advancements in differential cryptanalytic techniques that the algorithm could be proven vulnerable to collision attacks \cite{cryptoeprint:2005/010}, leading to the adoption of the still widely-used SHA2 algorithm.

The story of SHA2 is not unique in the world of cryptography, and is in fact also seen in PBKDF2 \cite{rfc2898}, Argon2 \cite{argon2}, WPA2 \cite{wpa2}, and many other such widely-used cryptographic algorithms that built upon a much weaker initial version.
The reason, too, is often the same: the introduction of brand new cryptographic paradigms are often accompanied by outdated or poorly-understood threat models. Only once the feasibility of the new paradigm is taken sufficiently seriously by the cryptography community do cryptanalytic techniques eventually catch up to the extent necessary for subsequent versions of the scheme to be robustly constructed.

The journey from NS23's MFKDF construction to MFKDF2 has clearly followed a similar pattern, offering valuable insights into the design of stateful cryptographic protocols. While the NS23 construction was a powerful proof-of-concept, the subsequent SBH24 cryptanalysis underscored the critical lesson that by forging the new category of ``stateful KDFs,'' NS23 also inadvertently made the static security analysis used by PBKDFs obsolete. The true security of such a system can only be understood by modeling its evolution over time under an active adversary, as formalized by our ESTMF technique. MFKDF2 is the result of applying this more rigorous, dynamic analysis, leading to a production-ready design that is not just patched, but fundamentally hardened.

\subsection{Stateful KDF Design Best Practices}

\subsubsection{What's In a Secret?}

Stateful KDFs challenge conventional assumptions about which values are and are not sensitive in cryptographic protocols, and require careful consideration of which values may actually leak information to an adversary or fall under adversarial control. Two helpful guidelines are as follows:

\begin{enumerate}
    \item \textbf{Scrutinize commutative properties that distribute secrecy.} Commutative operations like bitwise XORs blur the line between ``plaintext'' and ``key,'' leading to vulnerabilities like SBH24's ``HOTP Compromise Attack'' where a leaked ciphertext and plaintext can lead to the compromise of a key. In \S\ref{sec:analysis}, we found that the same issue can also be leveraged to allow one TOTP code to leak future TOTP codes due to the misuse of modular addition, a similar commutative operation. Such operations should be replaced with PRP-based encryption when possible.

    \item \textbf{Security parameters are themselves sensitive values}. While security parameters like ``key length'' or ``threshold'' aren't typically treated as sensitive in other cryptographic primitives, parameter tampering attacks like SBH24's ``Short Key Attack'' demonstrate that these values should be treated as attacker-controlled unless their integrity is otherwise protected (such as by MFKDF2's self-referential MAC). Enforcing sensible fixed security parameter values is one way to avoid the issue entirely.
\end{enumerate}

\subsubsection{The Dangers of Authentication}

The second suggestion listed above is a great example of how a lack of integrity protection can cause a KDF construction to fail by unintentionally placing inputs into adversarial control.
In other cases, however, the presence of an integrity mechanism can cause more harm than good by allowing attackers to ``check their work,'' performing brute-force attacks more efficiently than should otherwise be possible. KDF designers should equally consider both implicit and explicit integrity:

\begin{enumerate}[resume]
    \item \textbf{Scrutinize explicit integrity mechanisms.} In MFKDF, encrypting shares with an authenticated encryption mechanism like AES-GCM would have allowed authentication factors to be individually guessed and checked, violating the exponential entropy goals of the scheme. Counterintuitively, unauthenticated schemes like AES-ECB are generally preferable in stateful KDFs.

    \item \textbf{Scrutinize implicit integrity mechanisms.} Equally important and perhaps more pernicious are mechanisms in KDFs that bias plaintext values and therefore implicitly authenticate them. In NS23, biases in the data structure used to represent Shamir's secret shares allowed decrypted shares to be implicitly verified, resulting in the ``Share Format Attack.'' Thus, even when using unauthenticated encryption like AES-ECB, plaintext values must be indistinguishable from uniformly random.
\end{enumerate}

\subsubsection{Security Across Many Invocations}

The core finding of this paper is that the single-invocation threat model used to analyze the security of PBKDFs is insufficient for analyzing stateful multi-factor KDFs. We can derive two key insights directly from this premise:

\begin{enumerate}[resume]
    \item \textbf{Switching moduli can leak entropy.} As demonstrated by SBH24's ``HOTP Bias'' attack, converting uniformly random values in some range (e.g., $[0, 2^{31})$) to a non-compatible modulus (e.g., $[0, 10^6)$) causes the values in the second modulus to inherit a slight bias. Over time, these biases can slowly leak information to an adversary, meaning that they must be avoided or counteracted.

    \item \textbf{Oracles derive trust from liveness.} Our use of a timing oracle in \S\ref{subsec:timing_oracle} to simultaneously address the ``Dynamic Factor'' and ``HOTP Bias'' attacks of SBH24 demonstrates a powerful design pattern whereby an optional online 3rd-party participant adds significant value to the protocol while only being trusted to be live (not honest).
\end{enumerate}


\subsection{Stateful KDF Usage Best Practices}

Just as important as the secure design of KDFs is the secure implementation and use of KDFs in production systems. As the features and security properties of MFKDF2, MFCHF2, and MFDPG2 approach a state of production readiness, we also wish to take some time to discuss best practices for how these primitives should be utilized.

\begin{enumerate}
    \item \textbf{Enforce existing authentication policies.} Like the original constructions in NS23, MFKDF2 and MFCHF2 are flexible constructions that are agnostic to particular authentication policies, and don't impose specific factors or combinations thereof that should be used. Rather, these primitives should be used to take an existing, implicitly-enforced authentication policy and cause it to be enforced cryptographically to directly protect data.

    \item \textbf{Leverage defense-in-depth.} If PBKDFs are currently in use for encryption, then passwords are by definition a mandatory factor in the existing authentication policy. Rather than entirely replacing PBKDFs, leverage defense-in-depth by encrypting data with an MFKDF2 key on top of the existing PBKDF-based encryption. This approach has the dual benefit of not requiring existing data to be decrypted, and also guaranteeing that the use of MFKDF2 cannot weaken the overall security.

    \item \textbf{Guard key derivation behind authentication.} In practice, systems should use MFCHF2 and MFKDF2 together for authentication and encryption respectively (even with the same factors), protecting the MFKDF2 policy (and optional timing oracles) and behind an MFCHF2-based authentication process. In practice, this ensures that an adversary has already demonstrated active posession of all relevant authentication factors before they can begin to attack an MFKDF2-based key.
\end{enumerate}

\section{Future Work}

Though the ESTMF and MFKDF2 constructions presented in this paper provide a robust model for analyzing stateful KDFs and address all known vulnerabilities of the NS23 scheme, this work also raises several avenues for future research.

The ESTMF highlights that the entire class of state management vulnerabilities in NS23 stems from the existence of a public, mutable state. A truly ``stateless'' multi-factor KDF, as alluded to in SBH24, would be a paradigm shift, eliminating the need for cumbersome MAC-based state integrity mechanisms altogether.  
Research into constructions based on Oblivious Pseudorandom Functions could be promising, allowing a client to combine a factor with a server's secret to derive a key without either party learning the other's input. Similarly, techniques such as Function Secret Sharing, Secure Multi-Party Computation, or Fully Homomorphic Encryption could poteitally be investigated to create a stateless multi-factor KDF.
However, achieving this would almost certainly require a move towards interactive protocols, where a server-side secret contributes to the derivation without being stored in a client-side state, losing the ability to derive keys offline.

As the threat of cryptographically-relevant quantum computing grows, key management schemes must also be prepared for a post-quantum future.
While the core MFKDF2 construction presented in this paper relies almost entirely on symmetric cryptography and can therefore already reasonably called post-quantum, additional work may be required to upgrade factor constructions relying on asymmetric cryptography, such as the OOBA factors.
Furthermore, to maintain 256-bit security against a quantum adversary employing Grover's algorithm, the underlying symmetric key length must be doubled.
We therefore suggest that the fixed-length 256-bit MFKDF2 construction proposed in this paper be considered ``MFKDF2-256,'' and suggest the future development of an ``{MFKDF2-512}'' standard, which would be built upon 512-bit secret and intermediate values and a 512-bit stream cipher. 

In light of the new factors, features, and security improvements described in this paper, the robust security and flexibility of MFKDF2 make it suitable for a wide range of applications beyond its initial scope.
We propose future work investigating the use of MFKDF2 for OS-level full-disk encryption, allowing users to unlock their devices and encrypted volumes with a multi-factor policy (e.g., a password and a YubiKey), a significant improvement over the PBKDFs commonly used today.
MFKDF2 could also serve as a drop-in replacement for the password-based handshakes in protocols like WPA3's Simultaneous Authentication of Equals (SAE), enabling the creation of multi-factor protected Wi-Fi networks that allow for increased security and improved user experience.

Finally, while the proofs in this paper's appendices are constructed using standard cryptographic assumptions, the complexity of the MFKDF2 protocol makes it an ideal future candidate for formal verification. Such an analysis could provide even stronger, machine-checked guarantees about its security properties. Likewise, verifiable builds of the software ensures that clients run the correct algorithms.

\section{Conclusion}
\label{sec:conclusion}

The Multi-Factor Key Derivation Function proposed by NS23 had the compelling goal of incorporating multiple standard authentication factors directly into a client-side key derivation process, achieving exponential security gains over traditional PBKDFs \cite{NairSong2023}. Subsequent work showed that if realized, a secure MFKDF construction would offer quantifiable security and usability improvements in a wide variety of applications and settings \cite{nair2023mfdpgmultifactorauthenticatedpassword, NairSong2023_mfchf, NairSong2023_mfdpg, 10174998, 10.1145/3613904.3642464}.
However, subsequent cryptanalysis by SBH24 underscored a critical lesson for modern protocol design: for stateful cryptographic schemes, a static, single-invocation security analysis is insufficient \cite{Scarlata2024}. The true security of such a system can only be understood by modeling its evolution over time under an active adversary.

In this paper, we introduced the ESTMF as a novel technique to provide this necessary dynamic analysis for KDFs. By modeling a protocol as a state machine and formally defining the channels through which entropy can leak, ESTMF allowed us to systematically categorize the vulnerabilities in NS23 (including previously undiscovered vulnerabilities) and establish a rigorous security benchmark for a successor.

The result of this analysis is MFKDF2, a new construction that is fundamentally hardened against the classes of attacks that plagued its predecessor. By mandating the use of secure PRPs, enforcing state integrity with a self-referential MAC, and ensuring forward secrecy through stateful share regeneration, MFKDF2 is provably secure even within the stronger ESTMF model. Furtherm, we have shown that MFKDF2 can be extended with a rich ecosystem of modern authentication factors (including Passkeys), and new features like cryptographic upgradeability and threshold factor hints, demonstrating the power and flexibility of the MFKDF2 paradigm.

In the process of building the ESTMF, analyzing NS23, and constructing the hardened MFKDF2, we derived several generalizable best practices for the construction and use of KDFs that we hope will be insightful for the future users of MFKDF2 and builders of similar schemes.
Ultimately, it is our sincere hope that MFKDF2 follows in the footsteps of SHA2, PBKDF2, Argon2, and WPA2, and upholds the common trend whereby the second version of a cryptographic algorithm
builds upon the security lessons highlighted by a novel first release to achieve hardened security and widespread adoption.


\section*{Acknowledgements}

This work was conducted by Multifactor Research and was supported in part by the Zcash Foundation, the Hertz Foundation, Berkeley RDI.
Any opinions, findings, and conclusions or recommendations
expressed in this material are those of the authors and do not
necessarily reflect the views of the supporting entities.

\bibliographystyle{plain}
\bibliography{990_references}

\begin{thebibliography}{10}

\bibitem{1password_arch}
1Password.
\newblock {1Password} {Security} {Design}.

\bibitem{1PasswordPasskeysProduct}
{1Password}.
\newblock Sign in faster and more securely with passkeys.
\newblock Web Page, 2024.

\bibitem{credential_stuffing}
Akamai.
\newblock 2020 state of the internet.

\bibitem{ios}
Apple.
\newblock {iOS} {Security}, May 2012.

\bibitem{AppleHardwareSecurity}
{Apple Inc.}
\newblock Hardware security overview.
\newblock Apple Platform Security Guide, May 2024.

\bibitem{10.1007/978-3-031-91101-9_14}
Matilda Backendal, Sebastian Clermont, Marc Fischlin, and Felix G{\"u}nther.
\newblock Key derivation functions without a grain of salt.
\newblock In Serge Fehr and Pierre-Alain Fouque, editors, {\em Advances in Cryptology -- EUROCRYPT 2025}, pages 393--426, Cham, 2025. Springer Nature Switzerland.

\bibitem{10724844}
K~Balaji, Aruna~Sri Rongali, Ns~Archana, G.~Sethuraman, Jithesh~Mon Mullool, and R.~Karthi.
\newblock A critical analysis of key management techniques in applied cryptography.
\newblock In {\em 2024 15th International Conference on Computing Communication and Networking Technologies (ICCCNT)}, pages 1--5, 2024.

\bibitem{biokey}
Lucas Ballard, Seny Kamara, and Michael~K. Reiter.
\newblock The {Practical} {Subtleties} of {Biometric} {Key} {Generation}.

\bibitem{FIPS186-5}
Elaine Barker, Lily Chen, Sharon Keller, Allen Roginsky, Apostol Vassilev, and Richard Davis.
\newblock {Digital Signature Standard (DSS)}.
\newblock FIPS Publication 186-5, National Institute of Standards and Technology, February 2023.

\bibitem{argon2}
Alex Biryukov, Daniel Dinu, and Dmitry Khovratovich.
\newblock Argon2: the memory-hard function for password hashing and other applications.

\bibitem{cryptoeprint:2016/027}
Dan Boneh, Henry Corrigan-Gibbs, and Stuart Schechter.
\newblock Balloon hashing: A memory-hard function providing provable protection against sequential attacks.
\newblock Cryptology {ePrint} Archive, Paper 2016/027, 2016.

\bibitem{BonehShoup2020}
Dan Boneh and Victor Shoup.
\newblock A graduate course in applied cryptography, 2020.
\newblock \url{https://toc.cryptobook.us/}.

\bibitem{windows}
Elie Burzstein and Jean~Michel Picod.
\newblock Recovering {Windows} {Secrets} and {EFS} {Certiﬁcates} {Ofﬂine}.

\bibitem{CorriganGibbs2016}
Henry Corrigan-Gibbs, Dan Boneh, and Stuart~E. Schechter.
\newblock Balloon hashing: Provably space-hard hash functions with data-independent access patterns.
\newblock {\em IACR Cryptol. ePrint Arch.}, 2016:27, 2016.

\bibitem{dalskov_2fe_2020}
Anders Dalskov, Daniele Lain, Enis Ulqinaku, Kari Kostiainen, and Srdjan Capkun.
\newblock {2FE}: {Two}-{Factor} {Encryption} for {Cloud} {Storage}, October 2020.
\newblock arXiv:2010.14417 [cs].

\bibitem{dashlane_arch}
Dashlane.
\newblock Security white paper.
\newblock Technical report, Dashlane, March 2021.

\bibitem{DodisOstrovsky2008}
Yevgeniy Dodis, Rafail Ostrovsky, Leonid Reyzin, and Adam Smith.
\newblock Fuzzy extractors: How to generate strong keys from biometrics and other noisy data.
\newblock {\em SIAM Journal on Computing}, 38(1):97--139, 2008.

\bibitem{FIPS202}
Morris Dworkin.
\newblock {SHA-3 Standard: Permutation-Based Hash and Extendable-Output Functions}.
\newblock FIPS Publication 202, National Institute of Standards and Technology, August 2015.

\bibitem{cryptoeprint:2021/1522}
Ittay Eyal.
\newblock On cryptocurrency wallet design.
\newblock Cryptology ePrint Archive, Paper 2021/1522, 2021.
\newblock \url{https://eprint.iacr.org/2021/1522}.

\bibitem{chen_modular_2014}
Nils Fleischhacker, Mark Manulis, and Amir Azodi.
\newblock A {Modular} {Framework} for {Multi}-{Factor} {Authentication} and {Key} {Exchange}.
\newblock In Liqun Chen and Chris Mitchell, editors, {\em Security {Standardisation} {Research}}, volume 8893, pages 190--214. Springer International Publishing, Cham, 2014.
\newblock Series Title: Lecture Notes in Computer Science.

\bibitem{florencio_large_2006}
Dinei Florencio and Cormac Herley.
\newblock A {Large} {Scale} {Study} of {Web} {Password} {Habits}, November 2006.

\bibitem{10805375}
Qianwen Gao and Yichi Tu.
\newblock A blockchain wallet scheme with multi-factor authentication based on distributed system.
\newblock In {\em 2024 4th International Conference on Blockchain Technology and Information Security (ICBCTIS)}, pages 55--61, 2024.

\bibitem{grcGRCapossxA0xA0SQRLSecure}
GRC.
\newblock {G}{R}{C} --- {S}{Q}{R}{L} {S}ecure {Q}uick {R}eliable {L}ogin --- grc.com.
\newblock \url{https://www.grc.com/sqrl/sqrl.htm}.
\newblock [Accessed 26-08-2025].

\bibitem{hakimsecure}
Mr~BOUDJELABA Hakim and Mr~Belghit Lounes.
\newblock Secure key management mechanism in big data.

\bibitem{password_reuse}
Ameya Hanamsagar, Simon~S Woo, Christopher Kanich, and Jelena Mirkovic.
\newblock How {Users} {Choose} and {Reuse} {Passwords}.

\bibitem{HOANG2025104179}
Dinh~Linh Hoang, Thi~Luong Tran, and Van~Long Nguyen.
\newblock New proofs for pseudorandomness of hmac-based key derivation functions (rfc 5869).
\newblock {\em Journal of Information Security and Applications}, 93:104179, 2025.

\bibitem{ide2025personhood}
Ayae Ide and Tanusree Sharma.
\newblock Personhood credentials: Human-centered design recommendation balancing security, usability, and trust.
\newblock {\em arXiv preprint arXiv:2502.16375}, 2025.

\bibitem{ooba}
Ping Identity.
\newblock What is {Out}-of-{Band} {Authentication} ({OOBA})?

\bibitem{wpa}
IEEE.
\newblock Ieee standard for wireless lan medium access control (mac) and physical layer (phy) specifications.
\newblock {\em IEEE Std 802.11-1997}, 1997.

\bibitem{wpa2}
IEEE.
\newblock Ieee standard for information technology-telecommunications and information exchange between systems-local and metropolitan area networks-specific requirements-part 11: Wireless lan mac and phy specifications: Amendment 6: Mac security enhancements.
\newblock {\em IEEE Std 802.11i-2004}, pages 1--190, 2004.

\bibitem{cryptoeprint:2016/144}
Stanislaw Jarecki, Aggelos Kiayias, Hugo Krawczyk, and Jiayu Xu.
\newblock Highly-efficient and composable password-protected secret sharing (or: How to protect your bitcoin wallet online).
\newblock Cryptology ePrint Archive, Paper 2016/144, 2016.
\newblock \url{https://eprint.iacr.org/2016/144}.

\bibitem{rfc2898}
Burt Kaliski.
\newblock {PKCS} \#5: {Password}-{Based} {Cryptography} {Specification} {Version} 2.0.
\newblock Request for {Comments} RFC 2898, Internet Engineering Task Force, September 2000.
\newblock Num Pages: 34.

\bibitem{lastpass_arch}
LastPass.
\newblock Our {Zero}-{Knowledge} {Security} {Model}.

\bibitem{liu_multi-factor_2011}
Ying Liu, Fushan Wei, and Chuangui Ma.
\newblock Multi-{Factor} {Authenticated} {Key} {Exchange} {Protocol} in the {Three}-{Party} {Setting}.
\newblock In Xuejia Lai, Moti Yung, and Dongdai Lin, editors, {\em Information {Security} and {Cryptology}}, pages 255--267, Berlin, Heidelberg, 2011. Springer Berlin Heidelberg.

\bibitem{WindowsHello}
{Microsoft}.
\newblock Windows hello biometric security.
\newblock Microsoft Learn, October 2020.

\bibitem{MFKDFDocsEnveloping}
{Multifactor}.
\newblock Tutorial: Key enveloping.
\newblock MFKDF Documentation, 2024.

\bibitem{10174998}
Vivek Nair and Dawn Song.
\newblock Decentralizing custodial wallets with mfkdf.
\newblock In {\em 2023 IEEE International Conference on Blockchain and Cryptocurrency (ICBC)}, pages 1--9, 2023.

\bibitem{nair2023mfdpgmultifactorauthenticatedpassword}
Vivek Nair and Dawn Song.
\newblock Mfdpg: Multi-factor authenticated password management with zero stored secrets, 2023.

\bibitem{NairSong2023_mfdpg}
Vivek Nair and Dawn Song.
\newblock Mfdpg: Multi-factor authenticated password management with zero stored secrets, 2023.

\bibitem{10190544}
Vivek Nair and Dawn Song.
\newblock Multi-factor credential hashing for asymmetric brute-force attack resistance.
\newblock In {\em 2023 IEEE 8th European Symposium on Security and Privacy (EuroS\&P)}, pages 56--72, 2023.

\bibitem{NairSong2023_mfchf}
Vivek Nair and Dawn Song.
\newblock {M}ulti-{F}actor {C}redential {H}ashing for {A}symmetric {B}rute-{F}orce {A}ttack {R}esistance.
\newblock {\em {I}{E}{E}{E} {E}uropean {S}ymposium on {S}ecurity and {P}rivacy ({E}uro{S}\&{P})}, 2023.

\bibitem{NairSong2023}
Vivek Nair and Dawn Song.
\newblock {Multi-Factor} key derivation function ({{{{{MFKDF}}}}}) for fast, flexible, secure, \& practical key management.
\newblock In {\em 32nd USENIX Security Symposium (USENIX Security 23)}, pages 2097--2114, Anaheim, CA, August 2023. USENIX Association.

\bibitem{BIP0039}
Marek Palatinus, Pavol Rusnak, Aaron Voisine, and Sean Bowe.
\newblock {BIP-0039}: Mnemonic code for generating deterministic keys.
\newblock Bitcoin Improvement Proposals, September 2013.

\bibitem{pointcheval_multi-factor_2008}
David Pointcheval and Sébastien Zimmer.
\newblock Multi-factor {Authenticated} {Key} {Exchange}.
\newblock In Steven~M. Bellovin, Rosario Gennaro, Angelos Keromytis, and Moti Yung, editors, {\em Applied {Cryptography} and {Network} {Security}}, pages 277--295, Berlin, Heidelberg, 2008. Springer Berlin Heidelberg.

\bibitem{cryptoeprint:2005/010}
Vincent Rijmen and Elisabeth Oswald.
\newblock Update on {SHA}-1.
\newblock Cryptology {ePrint} Archive, Paper 2005/010, 2005.

\bibitem{Scarlata2024}
Matteo Scarlata, Matilda Backendal, and Miro Haller.
\newblock {MFKDF}: Multiple factors knocked down flat.
\newblock In {\em 33rd USENIX Security Symposium (USENIX Security 24)}, pages 4301--4318, Philadelphia, PA, August 2024. USENIX Association.

\bibitem{seo_construction_2018}
Minhye Seo, Jong~Hwan Park, Youngsam Kim, Sangrae Cho, Dong~Hoon Lee, and Jung~Yeon Hwang.
\newblock Construction of a {New} {Biometric}-{Based} {Key} {Derivation} {Function} and {Its} {Application}.
\newblock {\em Security and Communication Networks}, 2018:1--14, December 2018.

\bibitem{10.1145/3613904.3642464}
Tanusree Sharma, Vivek~C Nair, Henry Wang, Yang Wang, and Dawn Song.
\newblock “i can’t believe it’s not custodial!”: Usable trustless decentralized key management.
\newblock In {\em Proceedings of the 2024 CHI Conference on Human Factors in Computing Systems}, CHI '24, New York, NY, USA, 2024. Association for Computing Machinery.

\bibitem{cryptoeprint:2023/1785}
Yaqing Song, Yuan Zhang, Shiyu Li, Weijia Li, Zeqi Lai, and Qiang Tang.
\newblock There is always a way out! destruction-resistant key management: Formal definition and practical instantiation.
\newblock Cryptology {ePrint} Archive, Paper 2023/1785, 2023.

\bibitem{Soutar1998BiometricEU}
Colin Soutar, Danny Roberge, Alex Stoianov, Rene~M. Gilroy, and B.~V. K.~Vijaya Kumar.
\newblock Biometric encryption using image processing.
\newblock In {\em Electronic Imaging}, 1998.

\bibitem{rfc3174}
Cisco Systems.
\newblock Us secure hash algorithm 1 (sha1).
\newblock Request for {Comments} RFC 4226, Internet Engineering Task Force, December 2001.

\bibitem{usenixUSENIXSecurity}
USENIX.
\newblock {U}{S}{E}{N}{I}{X} {S}ecurity '23 {F}all {A}ccepted {P}apers --- usenix.org.
\newblock \url{https://www.usenix.org/conference/usenixsecurity23/fall-accepted-papers}.
\newblock [Accessed 24-08-2025].

\bibitem{6920371}
Mariano Luis~T. Uymatiao and William Emmanuel~S. Yu.
\newblock Time-based otp authentication via secure tunnel (toast): A mobile totp scheme using tls seed exchange and encrypted offline keystore.
\newblock In {\em 2014 4th IEEE International Conference on Information Science and Technology}, pages 225--229, 2014.

\bibitem{uzun_cryptographic_2021}
Erkam Uzun, Carter Yagemann, Simon Chung, Vladimir Kolesnikov, and Wenke Lee.
\newblock Cryptographic {Key} {Derivation} from {Biometric} {Inferences} for {Remote} {Authentication}.
\newblock In {\em Proceedings of the 2021 {ACM} {Asia} {Conference} on {Computer} and {Communications} {Security}}, {ASIA} {CCS} '21, pages 629--643, New York, NY, USA, May 2021. Association for Computing Machinery.

\bibitem{rfc4226}
Mountain View, David M'Raihi, Frank Hoornaert, David Naccache, Mihir Bellare, and Ohad Ranen.
\newblock {HOTP}: {An} {HMAC}-{Based} {One}-{Time} {Password} {Algorithm}.
\newblock Request for {Comments} RFC 4226, Internet Engineering Task Force, December 2005.
\newblock Num Pages: 37.

\bibitem{rfc6238}
Mountain View, Johan Rydell, Mingliang Pei, and Salah Machani.
\newblock {TOTP}: {Time}-{Based} {One}-{Time} {Password} {Algorithm}.
\newblock Request for {Comments} RFC 6238, Internet Engineering Task Force, May 2011.
\newblock Num Pages: 16.

\bibitem{webauthn_prf}
{W3C Web Authentication Working Group}.
\newblock Web authentication: An api for accessing public key credentials level 3.
\newblock W3C Proposed Recommendation, July 2024.

\bibitem{10.5555/1251421.1251435}
Alma Whitten and J.~D. Tygar.
\newblock Why johnny can't encrypt: a usability evaluation of pgp 5.0.
\newblock In {\em Proceedings of the 8th Conference on USENIX Security Symposium - Volume 8}, SSYM'99, page~14, USA, 1999. USENIX Association.

\end{thebibliography}

\appendix

\clearpage

\section{Preliminaries}
We first recall the core components of the NS23 MFKDF model \cite{NairSong2023}. The framework consists of two main components: \emph{factor constructions}, which derive static secret material from various authentication methods, and \emph{KDF constructions}, which combine that material to derive a key.

\begin{definition}
A \emph{factor} $F_i$ is a cryptographic primitive defined by a secret component, the \emph{factor material} ($\sigma_i$), and a public component, the \emph{public state} ($\beta_i$).
\end{definition}

For example, for a password factor, the factor material $\sigma_i$ is the password itself, while its public state $\beta_i$ might be a static identifier. For a TOTP factor, $\sigma_i$ is the shared secret key, and $\beta_i$ contains public information like the current counter, which changes over time.

\begin{definition}
A \emph{factor construction} is a set of stateful algorithms, \texttt{Setup} and \texttt{Derive}, that manage a factor's lifecycle.
\begin{itemize}
    \item The \texttt{Setup} algorithm takes the initial factor material $\sigma_i$ and produces the initial public state $\beta_{i,0}$ and the static source key material $\kappa_{i}$.
    \item The \texttt{Derive} algorithm uses a \emph{factor witness} $W_{i,j}$ and the current public state $\beta_{i,j}$ to output the \emph{source key material} $\kappa_{i}$ and the next state $\beta_{i,j+1}$. For dynamic factors requiring a key feedback mechanism, this is a two-stage process. First, the source key material is reconstructed\footnote{Note that in \cite{NairSong2023}, the source key material $\kappa$ is referred to as ``factor material'' and the public state $\beta$ as ``factor parameters'' (denoted $\sigma$ and $\alpha$, respectively). We adopt this separated notation for clarity, as not all factor constructions use the factor material $\sigma$ as the source key material $\kappa$.}:
    \[ \kappa_{i} \leftarrow \texttt{Derive}_1(W_{i,j}, \beta_{i,j}) \]
    After the final MFKDF \emph{derived key} $K$ is derived, the second stage updates the state:
    \[ \beta_{i,j+1} \leftarrow \texttt{Derive}_2(K, \kappa_{i}, \beta_{i,j}) \]
\end{itemize}
\end{definition}

First, note that $\sigma$ and $\kappa$ are \emph{static} material and do not ever receive a state-step $j$ subscript. A motivating example that distinguishes $\sigma$ and $\kappa$ is the TOTP factor construction. The long-term secret shared with the user's authenticator app is the factor material $\sigma$ even though secret itself is never used directly in the main KDF. Instead, during \texttt{Setup}, a random, static source key material $\kappa$ is generated. The factor's public state $\beta$ then stores an encrypted version of $\sigma$ (using the key feedback mechanism) and public helper data (an "offset"). During \texttt{Derive}, the user's 6-digit TOTP code (the witness $W$) is combined with this public helper data to reconstruct the static $\kappa$. Thus, $\sigma$ is the underlying secret that ``powers'' the factor, while $\kappa$ is the consistent value that the factor contributes to the final key derivation.

To enable features like account recovery and flexible access control, NS23 introduces a thresholding mechanism based on Shamir's Secret Sharing. In this model, a \emph{master secret} $M$ is split into $n$ \emph{shares}, denoted $\{s_i\}$, which are then encrypted to create \emph{encrypted shares} $\{c_{i}\}$ and stored in the public state. Any $t$ of these shares can be used to reconstruct the master secret. The following \texttt{Policy} definition provides the most general version of this concept, allowing for arbitrarily complex combinations of factors beyond simple thresholds via "key stacking," as described in the original implementation. In the the $n$ of $n$ case, the master secret $M$ is deterministically derived by concatenating the salted hashes of the source key material $\{\kappa_{i}\}$ from all $n$ factors.

\begin{definition}
Given a set of factors $\mathcal{F}=\{F_1, F_2, \dots, F_n\}$, a policy $P$ is a set of all allowable factor combinations ($P=\{C_1, C_2, \dots\}$) that can be used to derive the final key $K$.
\end{definition}

The collection of all public parameters and encrypted secrets for a given policy forms the complete public state of the MFKDF instance at a specific point in time.

\begin{definition}
The complete set of public information for an MFKDF instance after the $i$-th derivation is its \emph{public state}, denoted $\mathcal{B}_i$. This state is a tuple containing the MFKDF construction parameters (e.g., policy, salt), the set of encrypted shares $\{c_{s_j}\}$, and the collection of all individual public factor states $\{\beta_{1,i}, \beta_{2,i}, \dots, \beta_{n,i}\}$.
\end{definition}

In addition to setting up and deriving keys, the NS23 paradigm introduced a mechanism for replacing lost or compromised factors without changing the final key, a process called reconstitution.

\begin{definition}
The \texttt{Recover} algorithm allows a lost factor $F_j$ to be replaced with a new factor $F'_i$. Let $\mathcal{W}_C = \{W_{i,j,k} \mid k \in C\}$ be a set of valid witnesses for a combination of factors $C$ that satisfies the policy, let $\sigma'_i$ be the material for the new factor $F'_i$, and let $\mathcal{B}_j$ be the current state. The state transition is defined as:
\[ \mathcal{B}_{j+1} \leftarrow \texttt{Recover}(W_C, \mathcal{B}_j, i, \sigma'_i) \]
\end{definition}

Together, the \texttt{Setup}, \texttt{Derive}, and \texttt{Recover} algorithms act as the state transition functions for the MFKDF state machine, noting that \texttt{Setup} is always run for initialization. The ability to update the state is what enables the MFKDF paradigm to support dynamic factors and account recovery. 

Lastly, we will have the MFKDF state carry an authentication tag $\mathcal{T}$. This can be used to guarantee the integrity of the state for a user of an MFKDF scheme.

\clearpage

\section{Algorithms}
\label{sec:algorithms}

This section provides the formal algorithmic construction for MFKDF2. We begin by defining the notation used throughout the specifications.

\subsection{MFKDF2 Construction}
\label{sec:mfkdf2_construction}

\subsubsection{MFKDF2 With No Secret Sharing}

Algorithm \ref{alg:mfkdf2} specifies the simplest construction without secret sharing. It improves upon the NS23 equivalent by introducing two critical security features. First, it adds a unique, per-factor salt ($\mathit{salt}_{i}$) to each factor's source key material before combination, which prevents the "Factor Fungibility" attack. Second, the entire public state is authenticated with a self-referential MAC tag ($\mathcal{T}$) to provide integrity and prevent tampering attacks.

\begin{algorithm}[htbp]
\caption{MFKDF2 $n$ of $n$ Construction}
\label{alg:mfkdf2}
\begin{algorithmic}[1]
\small
\Require Let $\mathit{KDF}$ be Argon2id, $\mathit{MAC}$ be HMAC-SHA256
\Function{Setup}{$\{\sigma_i\}_{i=1}^n$}
    \State $\mathit{salt}_K \gets \{0,1\}^{256}$
    \State $\{\beta_{i,0}\}_{i=1}^n, \{\kappa_{i}\}_{i=1}^n, \{\mathit{salt}_{i}\}_{i=1}^n \gets \text{empty lists}$
    \ForAll{$i \in \{1, \dots, n\}$}
        \State $(\beta_{i,0}, \kappa_{i}) \gets F_i.\texttt{Setup}(\sigma_i)$
        \State $\mathit{salt}_{i} \gets \{0,1\}^{256}$ \Comment{Generate per-factor salt}
    \EndFor
    \State $M \gets \bigodot_{i=1}^n \operatorname{Hash}(\kappa_{i} \ || \ \mathit{salt}_{i})$ \Comment{Combine salted factors}
    \State $K_0 \gets \mathit{KDF}(M, \mathit{salt}_K)$ \Comment{Bootstrap key}
    \ForAll{$i \in \{1, \dots, n\}$}
        \State $\beta_{i,0} \gets F_i.\texttt{Derive}_2(K_0, \kappa_{i}, \beta_{i,0})$ \Comment{Key feedback}
    \EndFor
    \State $\mathcal{B}_0 \gets (\mathit{salt}_K, \{\mathit{salt}_{i}\}_{i=1}^n, \{\beta_{i,0}\}_{i=1}^n)$
    \State $\mathcal{T}_0 \gets \mathit{MAC}(K_0, \mathcal{B}_0)$
    \State \Return $((\mathcal{B}_0, \mathcal{T}_0), K_0)$
\EndFunction

\Function{Derive}{$\{W_{i,j}\}_{i=1}^n, (\mathcal{B}_j, \mathcal{T}_j)$}
    \State $(\mathit{salt}_K, \{\mathit{salt}_{i}\}_{i=1}^n, \{\beta_{i,j}\}_{i=1}^n) \gets \mathcal{B}_j$
    \State $\{\kappa_{i}\}_{i=1}^n \gets \text{empty list}$
    \ForAll{$i \in \{1, \dots, n\}$}
        \State $\kappa_{i} \gets F_i.\texttt{Derive}_1(W_{i,j}, \beta_{i,j})$
    \EndFor
    \State $M' \gets \bigodot_{i=1}^n \operatorname{Hash}(\kappa_{i} \ || \ \mathit{salt}_{i})$
    \State $K' \gets \mathit{KDF}(M', \mathit{salt}_K)$ \Comment{Derive candidate key}
    \State $\mathcal{T}' \gets \mathit{MAC}(K', \mathcal{B}_j)$ \Comment{Self-referential integrity check}
    \If{$\mathcal{T}' \neq \mathcal{T}_j$} \Return $\bot$ \EndIf
    \State $\{\beta_{i,j+1}\}_{i=1}^n \gets \{\beta_{i,j}\}_{i=1}^n$
    \ForAll{$i \in \{1, \dots, n\}$}
        \State $\beta_{i,j+1} \gets F_i.\texttt{Derive}_2(K', \kappa_{i}, \beta_{i,j})$
    \EndFor
    \State $\mathcal{B}_{j+1} \gets (\mathit{salt}_K, \{\mathit{salt}_{i}\}_{i=1}^n, \{\beta_{i,j+1}\}_{i=1}^n)$
    \State $\mathcal{T}_{j+1} \gets \mathit{MAC}(K', \mathcal{B}_{j+1})$
    \State \Return $((\mathcal{B}_{j+1}, \mathcal{T}_{j+1}), K')$
\EndFunction
\end{algorithmic}
\end{algorithm}

\subsubsection{Threshold MFKDF Construction}
\label{sec:threshold_mfkdf_construction}

Algorithm \ref{alg:mfkdf2_threshold} specifies the more complex t-of-n threshold construction. Unlike the n-of-n version which combines all factors, this construction uses Shamir's Secret Sharing to enable policy-based derivation and recovery where only a subset of factors is required. This is achieved by splitting a randomly generated master secret $\kappa$ into $n$ shares, encrypting each share with a key derived from its corresponding factor, and storing the ciphertexts in the public state.

\begin{algorithm}[htbp]
\caption{MFKDF2 Threshold Construction}
\label{alg:mfkdf2_threshold}
\begin{algorithmic}[1]
\small
\Require Let $\mathit{KDF}$ be Argon2id, $\mathit{MAC}$ be HMAC-SHA256.
\Require Let $(\mathit{Share},\mathit{Comb})$ be SSS over $\operatorname{GF}(256)$.
\Require Let $\mathit{PRP}$ be AES-256.

\Function{Setup}{$t, \{\sigma_i\}_{i=1}^n$}
    \State $M \gets \{0,1\}^{256}$ \Comment{Generate master secret}
    \State $\mathit{salt}_K \gets \{0,1\}^{256}$
    \State $\{s_i\}_{i=1}^n \gets \mathit{Share}(M, t, n)$
    \State $\{\beta_{i,0}\}_{i=1}^n, \{\kappa_{i}\}_{i=1}^n \gets \text{empty lists}$
    \ForAll{$i \in \{1, \dots, n\}$}
        \State $(\beta_{i,0}, \kappa_{i}) \gets F_i.\texttt{Setup}(\sigma_i)$
    \EndFor
    \State $K_0 \gets \mathit{KDF}(M, \mathit{salt}_K)$ \Comment{Bootstrap key for initial MAC}
    \State $\{c_{i}\}_{i=1}^n \gets \text{empty list}$
    \ForAll{$i \in \{1, \dots, n\}$}
        \State $k_i \gets \operatorname{HKDF}(\kappa_{i})$ \Comment{Derive share key}
        \State $c_{i} \gets \mathit{PRP}_{k_i}(s_i)$
    \EndFor
    \State $\mathcal{B}_0 \gets (t, n, \mathit{salt}_K, \{\beta_{i,0}\}_{i=1}^n, \{c_{i}\}_{i=1}^n)$
    \State $\mathcal{T}_0 \gets \mathit{MAC}(K_0, \mathcal{B}_0)$
    \State \Return $((\mathcal{B}_0, \mathcal{T}_0), K_0)$
\EndFunction

\Function{Derive}{$\{W_{i,j}\}_{i \in C}, (\mathcal{B}_j, \mathcal{T}_j)$}
    \State $(t, n, \mathit{salt}_K, \{\beta_{i,j}\}_{i=1}^n, \{c_{i}\}_{i=1}^n) \gets \mathcal{B}_i$
    \State \If{$|C| < t$} \Return $\bot$ \EndIf
    \State $\{\kappa_{i}\}_{i \in C}, \{s_i\}_{i \in C} \gets \text{empty lists}$
    \ForAll{$i \in C$}
        \State $\kappa_{i} \gets F_i.\texttt{Derive}_1(W_{i,j}, \beta_{i,j})$
        \State $k_i \gets \operatorname{HKDF}(\kappa_{i})$
        \State $s_i \gets \mathit{PRP}^{-1}_{k_i}(c_{i})$
    \EndFor
    \State $M' \gets \mathit{Comb}(\{s_i\}_{i \in C})$
    \State $K' \gets \mathit{KDF}(M', \mathit{salt}_K)$ \Comment{Derive candidate key}
    \State $\mathcal{T}' \gets \mathit{MAC}(K', \mathcal{B}_j)$ \Comment{Self-referential integrity check}
    \If{$\mathcal{T}' \neq \mathcal{T}_j$} \Return $\bot$ \EndIf
    \State $\{\beta_{i,j+1}\}_{i=1}^n \gets \{\beta_{i,j}\}_{i=1}^n$
    \ForAll{$i \in C$}
        \State $\beta_{i,j+1} \gets F_i.\texttt{Derive}_2(K', \kappa_{i}, \beta_{i,j})$
    \EndFor
    \State $\mathcal{B}_{j+1} \gets (t, n, \mathit{salt}_K, \{\beta_{i,j+1}\}_{i=1}^n, \{c_{i}\}_{i=1}^n)$
    \State $\mathcal{T}_{j+1} \gets \mathit{MAC}(K', \mathcal{B}_{j+1})$
    \State \Return $((\mathcal{B}_{j+1}, \mathcal{T}_{j+1}), K')$
\EndFunction
\end{algorithmic}
\end{algorithm}

\subsection{New Factor Constructions}
\label{app:factor_constructions}

This section provides the formal algorithms for the new factors supported by MFKDF2.

\subsection{New Factor Constructions}
\label{sec:factor_constructions}

This section provides the formal algorithms for the new factors supported by MFKDF2.

\begin{algorithm}[H]
\caption{Factor Construction for Passkeys (WebAuthn-PRF)}
\label{alg:mfkdf2_passkey_factor}
\begin{algorithmic}[1]
\Require Let $\operatorname{WebAuthn.getPRF}$ be the WebAuthn PRF extension call.
\Function{Setup}{$\sigma_i$} \Comment{$\sigma_i$ contains the WebAuthn credentialId}
    \State $\mathit{credentialId} \gets \sigma_i$
    \State $\mathit{chal}_{i,0} \gets \{0,1\}^{256}$ \Comment{Generate initial random challenge}
    \State $\beta_{i,0} \gets (\mathit{credentialId}, \mathit{chal}_{i,0})$
    \State $\kappa_{i} \gets \bot$ \Comment{Source key material is derived dynamically}
    \State \Return $(\beta_{i,0}, \kappa_{i})$
\EndFunction

\Function{Derive$_1$}{$W_{i,j}, \beta_{i,j}$} \Comment{$W_{i,j}$ is the PRF output}
    \State $\kappa_{i} \gets W_{i,j}$
    \State \Return $\kappa_{i}$
\EndFunction

\Function{Derive$_2$}{$K', \kappa_{i}, \beta_{i,j}$}
    \State $(\mathit{credentialId}, \mathit{chal}_{i,j}) \gets \beta_{i,j}$
    \State $\mathit{chal}_{i,j+1} \gets \{0,1\}^{256}$ \Comment{Generate next random challenge}
    \State $\beta_{i,j+1} \gets (\mathit{credentialId}, \mathit{chal}_{i,j+1})$
    \State \Return $\beta_{i,j+1}$
\EndFunction
\end{algorithmic}
\end{algorithm}

Algorithm \ref{alg:mfkdf2_passkey_factor} shows the MFKDF2 factor construction for Passkeys, which leverages the WebAuthn PRF extension \cite{webauthn_prf}. This provides a modern, asymmetric challenge-response mechanism. The public state $\beta_{i,j}$ contains the credential ID and a random challenge $c_{i,j}$. The user's authenticator computes the PRF output over this challenge, which serves as the witness $W_{i,j}$ and is directly used as the source key material $\kappa_{i}$. A new random challenge is generated in \texttt{Derive}$_2$ to ensure freshness for the next derivation.

\subsubsection{Fuzzy Encryption Factor}

\begin{algorithm}[htbp]
\caption{Factor Construction for Fuzzy Encryption}
\label{alg:mfkdf2_fuzzy_factor}
\begin{algorithmic}[1]
\Require Let $(\operatorname{FE.Gen}, \operatorname{FE.Rep})$ be a Fuzzy Extractor scheme.
\Function{Setup}{$\sigma_i$} \Comment{$\sigma_i$ is the initial high-quality sample $w$}
    \State $w \gets \sigma_i$
    \State $(\kappa_{i}, P) \gets \operatorname{FE.Gen}(w)$ \Comment{Generate key and public helper data}
    \State $\beta_{i,0} \gets P$
    \State \Return $(\beta_{i,0}, \kappa_{i})$
\EndFunction

\Function{Derive$_1$}{$W_{i,j}, \beta_{i,j}$} \Comment{$W_{i,j}$ is the new noisy sample $w'$}
    \State $w' \gets W_{i,j}$
    \State $P \gets \beta_{i,j}$
    \State $\kappa_{i} \gets \operatorname{FE.Rep}(w', P)$ \Comment{Reproduce key from noisy sample}
    \State \Return $\kappa_{i}$
\EndFunction

\Function{Derive$_2$}{$K', \kappa_{i}, \beta_{i,j}$}
    \State $\beta_{i,j+1} \gets \beta_{i,j}$ \Comment{Helper data is static and does not require key feedback}
    \State \Return $\beta_{i,j}$
\EndFunction
\end{algorithmic}
\end{algorithm}

Algorithm \ref{alg:mfkdf2_fuzzy_factor} specifies the construction for noise-tolerant factors using a Fuzzy Extractor \cite{DodisOstrovsky2008}. During \texttt{Setup}, the initial high-quality sample (e.g., a biometric scan) is used to generate a stable, high-entropy source key material $\kappa_{F_j}$ and public helper data $P$, which is stored as the public state $\beta_j$. During \texttt{Derive}, a new, noisy sample is provided as the witness, and the \texttt{Rep} algorithm uses it along with the helper data to reliably reconstruct the original $\kappa_{F_j}$. This construction is stateless after initialization, as the helper data does not need to be updated.

\subsubsection{SQRL-based Factor}

\begin{algorithm}[htbp]
\caption{Factor Construction for SQRL}
\label{alg:mfkdf2_sqrl_factor}
\begin{algorithmic}[1]
\Require Let $\operatorname{Sign}$ be a deterministic digital signature algorithm (e.g., EdDSA).
\Function{Setup}{$\sigma_i$} \Comment{$\sigma_i$ is the mobile device's public key $pk$}
    \State $pk \gets \sigma_i$
    \State $\mathit{chal}_{i,0} \gets \{0,1\}^{256}$ \Comment{Generate initial random challenge}
    \State $\beta_{i,0} \gets (pk, \mathit{chal}_{i,0})$
    \State $\kappa_{i} \gets \bot$
    \State \Return $(\beta_{i,0}, \kappa_{i})$
\EndFunction

\Function{Derive$_1$}{$W_{i,j}, \beta_{i,j}$} \Comment{$W_{i,j}$ is the signature over the challenge}
    \State $(pk , \mathit{chal}_{i,j}) \gets \beta_{i,j}$
    \State \Comment{Client verifies signature $W_{i,j}$ over challenge $\mathit{chal}_{i,j}$ using $pk$}
    \State $\kappa_{i} \gets W_{i,j}$
    \State \Return $\kappa_{i}$
\EndFunction

\Function{Derive$_2$}{$K', \kappa_{i}, \beta_{i,j}$}
    \State $(pk, \mathit{chal}_{i,j}) \gets \beta_{i,j}$
    \State $\mathit{chal}_{i,j+1} \gets \{0,1\}^{256}$ \Comment{Generate next random challenge}
    \State $\beta_{i,j+1} \gets (pk, \mathit{chal}_{i,j+1})$
    \State \Return $\beta_{i,j+1}$
\EndFunction
\end{algorithmic}
\end{algorithm}

Algorithm \ref{alg:mfkdf2_sqrl_factor} specifies the construction for a QR code-based factor, adapting the challenge-response flow of a protocol like SQRL. During \texttt{Setup}, the public key of the user's mobile device is stored in the public state $\beta_j$ alongside an initial random challenge. To derive, the client displays the current challenge as a QR code. The user scans it, their mobile device signs the challenge, and the resulting signature serves as the witness $W_{j,i}$ and the source key material $\kappa_{F_j}$. A new challenge is generated for the next iteration to ensure freshness.

\subsubsection{Push Notification Factor}

\begin{algorithm}[H]
\caption{Factor Construction for Push Notifications}
\label{alg:mfkdf2_push_factor}
\begin{algorithmic}[1]
\Require Let $\operatorname{Sign}$ be a deterministic digital signature algorithm (e.g., EdDSA).
\Function{Setup}{$\sigma_i$} \Comment{$\sigma_i$ is the mobile device's public key $pk$}
    \State $pk \gets \sigma_i$
    \State $\mathit{chal}_{i,0} \gets \{0,1\}^{256}$ \Comment{Generate initial random challenge}
    \State $\beta_{i,0} \gets (pk, \mathit{chal}_{i,0})$
    \State $\kappa_{i} \gets \bot$
    \State \Return $(\beta_{i,0}, \kappa_{i})$
\EndFunction

\Function{Derive$_1$}{$W_{i,j}, \beta_{i,j}$} \Comment{$W_{i,j}$ is the signature over the challenge}
    \State $(pk , \mathit{chal}_{i,j}) \gets \beta_{i,j}$
    \State \Comment{Client verifies signature $W_{i,j}$ over challenge $\mathit{chal}_{i,j}$ using $pk$}
    \State $\kappa_{i} \gets W_{i,j}$
    \State \Return $\kappa_{i}$
\EndFunction

\Function{Derive$_2$}{$K', \kappa_{i}, \beta_{i,j}$}
    \State $(pk, \mathit{chal}_{i,j}) \gets \beta_{i,j}$
    \State $\mathit{chal}_{i,j+1} \gets \{0,1\}^{256}$ \Comment{Generate next random challenge}
    \State $\beta_{i,j+1} \gets (pk, \mathit{chal}_{i,j+1})$
    \State \Return $\beta_{i,j+1}$
\EndFunction
\end{algorithmic}
\end{algorithm}

Algorithm \ref{alg:mfkdf2_push_factor} specifies the construction for a push notification-based factor. The cryptographic flow is identical to the QR code factor, with the primary difference being the transport mechanism for the challenge. Instead of being displayed as a QR code, the challenge $c_{j,i}$ is sent to the user's registered mobile device via a push notification service. The user approves the request, and their device signs the challenge, returning the signature as the witness $W_{j,i}$, which serves as the source key material $\kappa_{F_j}$.

\subsubsection{Biometrics and Secure Enclaves}
\begin{algorithm}[H]
\caption{Factor Construction for Enclave-Based Factors}
\label{alg:mfkdf2_enclave_factor}
\begin{algorithmic}[1]
\Require Let $\operatorname{Enclave.GenAndBind}$ and $\operatorname{Enclave.Release}$ be TEE API calls.
\Function{Setup}{$\sigma_i$} \Comment{$\sigma_i$ is an identifier for the binding condition}
    \State $\mathit{secretId} \gets \operatorname{Enclave.GenAndBind}(\sigma_i)$ \Comment{Enclave binds new secret to condition}
    \State $\beta_{i,0} \gets \mathit{secretId}$
    \State $\kappa_{i} \gets \bot$
    \State \Return $(\beta_{i,0}, \kappa_{i})$
\EndFunction

\Function{Derive$_1$}{$W_{i,j}, \beta_{i,j}$} \Comment{$W_{i,j}$ is the secret released by the enclave}
    \State $\kappa_{i} \gets W_{i,j}$
    \State \Return $\kappa_{i}$
\EndFunction

\Function{Derive$_2$}{$K', \kappa_{i}, \beta_{i,j}$}
    \State $\beta_{i,j+1} \gets \beta_{i,j}$ \Comment{Enclave-bound secret ID is static}
    \State \Return $\beta_{i,j+1}$
\EndFunction
\end{algorithmic}
\end{algorithm}

Algorithm \ref{alg:mfkdf2_enclave_factor} provides a general construction for factors backed by a Secure Enclave or Trusted Execution Environment (TEE). During \texttt{Setup}, the enclave is instructed to generate a high-entropy secret and cryptographically bind it to a specific condition (e.g., a successful biometric match, or a request from a specific application). The enclave returns a public handle, $\mathit{secretId}$, which is stored as the public state $\beta_j$. To derive, the client requests the secret associated with $\mathit{secretId}$. The enclave enforces the binding condition, and upon success, releases the secret, which serves as the witness $W_{j,i}$ and the source key material $\kappa_{F_j}$.

\subsubsection{Proximity Based Factors}

\begin{algorithm}[htbp]
\caption{Factor Construction for Proximity (NFC/RFID)}
\label{alg:mfkdf2_nfc_factor}
\begin{algorithmic}[1]
\Require Let $\operatorname{HMAC}$ be a secure MAC function (e.g., HMAC-SHA256).
\Function{Setup}{$\sigma_i$} \Comment{$\sigma_i$ is the secret key $K_{tag}$ stored on the tag}
    \State $K_{tag} \gets \sigma_i$
    \State $\mathit{chal}_{i,0} \gets \{0,1\}^{256}$ \Comment{Generate initial random challenge}
    \State $\beta_{i,0} \gets \mathit{chal}_{i,0}$
    \State $\kappa_{i} \gets \bot$
    \State \Return $(\beta_{i,0}, \kappa_{i})$
\EndFunction

\Function{Derive$_1$}{$W_{i,j}, \beta_{i,j}$} \Comment{$W_{i,j}$ is the HMAC response from the tag}
    \State $\mathit{chal}_{i,j} \gets \beta_{i,j}$
    \State \Comment{Client has received $W_{i,j}$ from tag in response to challenge $\mathit{chal}_{i,j}$}
    \State $\kappa_{i} \gets W_{i,j}$
    \State \Return $\kappa_{i}$
\EndFunction

\Function{Derive$_2$}{$K', \kappa_{i}, \beta_{i,j}$}
    \State $\mathit{chal}_{i,j+1} \gets \{0,1\}^{256}$ \Comment{Generate next random challenge}
    \State $\beta_{i,j+1} \gets \mathit{chal}_{i,j+1}$
    \State \Return $\beta_{i,j+1}$
\EndFunction
\end{algorithmic}
\end{algorithm}

Algorithm \ref{alg:mfkdf2_nfc_factor} provides a general construction for proximity-based factors like NFC and RFID. During \texttt{Setup}, a high-entropy secret key is provisioned onto the physical tag, and an initial random challenge is stored in the public state $\beta_j$. To derive, the client (e.g., a smartphone or RFID reader) transmits the current challenge $c_{j,i}$ to the tag. The tag computes a deterministic HMAC of the challenge with its secret key, and transmits the result back. This result serves as the witness $W_{j,i}$ and the source key material $\kappa_{F_j}$. A new challenge is generated for the next iteration to ensure freshness and prevent replay attacks.

\begin{algorithm}[H]
\caption{Factor Construction for OpenID Connect}
\label{alg:mfkdf2_oidc_factor}
\begin{algorithmic}[1]
\Require Let $\operatorname{OIDC.Auth}$ be the OIDC authentication flow.
\Function{Setup}{$\sigma_i$} \Comment{$\sigma_i$ contains provider URL and user's \texttt{sub}}
    \State $(\mathit{issuer}, \mathit{sub}) \gets \sigma_i$
    \State $\mathit{chal}_{i,0} \gets \{0,1\}^{256}$ \Comment{Generate initial random challenge}
    \State $\beta_{i,0} \gets (\mathit{issuer}, \mathit{sub}, \mathit{chal}_{i,0})$
    \State $\kappa_{i} \gets \bot$
    \State \Return $(\beta_{i,0}, \kappa_{i})$
\EndFunction

\Function{Derive$_1$}{$W_{i,j}, \beta_{i,j}$} \Comment{$W_{i,j}$ is the signature of the \texttt{id\_token}}
    \State $(\_ , \_ , \mathit{chal}_{i,j}) \gets \beta_{i,j}$
    \State \Comment{Client verifies the \texttt{id\_token} signature and that it contains $\mathit{chal}_{i,j}$}
    \State $\kappa_{i} \gets W_{i,j}$
    \State \Return $\kappa_{i}$
\EndFunction

\Function{Derive$_2$}{$K', \kappa_{i}, \beta_{i,j}$}
    \State $(\mathit{issuer}, \mathit{sub}, \_) \gets \beta_{i,j}$
    \State $\mathit{chal}_{i,j+1} \gets \{0,1\}^{256}$ \Comment{Generate next random challenge}
    \State $\beta_{i,j+1} \gets (\mathit{issuer}, \mathit{sub}, \mathit{chal}_{i,j+1})$
    \State \Return $\beta_{i,j+1}$
\EndFunction
\end{algorithmic}
\end{algorithm}

Algorithm \ref{alg:mfkdf2_oidc_factor} specifies the construction for an OpenID Connect (OIDC) factor. During \texttt{Setup}, the OIDC provider's details and the user's subject identifier are stored in the public state $\beta_j$ along with a random \texttt{nonce}. To derive, the client initiates an OIDC flow, including the current \texttt{nonce} as a challenge. The provider returns a signed JWT (\texttt{id\_token}) containing the nonce. The signature of this token serves as the witness $W_{j,i}$ and the source key material $\kappa_{F_j}$. A new nonce is generated for the next derivation to ensure freshness.

\section{Security Proofs for MFKDF2}
\label{sec:proofs}

This section provides the formal security analysis for the MFKDF2 construction. We prove that MFKDF2 is secure within the Entropy State Transition Modeling Framework (ESTMF) by demonstrating that it satisfies each of the framework's core principles. We begin by proving the security of the foundational mechanisms that mitigate the vulnerabilities identified in NS23.

\subsection{Integrity against Adversarial Flow}
ESTMF requires a scheme to be secure against the flow of malicious information from an adversary into the protocol's logic. MFKDF2 achieves this through a self-referential MAC that guarantees state integrity.

\begin{proposition}
\label{prop:state_integrity}
    Let the MFKDF2 public state be a tuple $(\mathcal{B}, \mathcal{T})$. The protocol is secure against state tampering if a client, after deriving a candidate key $K'$ from $B$, accepts the state as authentic only upon successful verification that $\mathcal{T} = \operatorname{MAC}(K', \mathcal{B})$.
\end{proposition}

\begin{proof}
    Denote the authentic state by $(\mathcal{B}, \mathcal{T})$ with $\mathcal{T} = \operatorname{MAC}(K, \mathcal{B})$ for the correct key $K$. Let an attacker produce $\mathcal{B}' \neq \mathcal{B}$, but note that the attacker must produce a tuple $(\mathcal{B}', \mathcal{T}')$ such that $\mathcal{T}'=\mathcal{T}$ which implies $\operatorname{MAC}(K', \mathcal{B}')=\mathcal{T}$ for some $K'$. However, due to the \emph{existential unforgeability under a chosen-message attack (EUF-CMA)} property of a secure MAC, the probability of $\mathcal{T}'=\mathcal{T}$ is negligible even if $K'=K$, therefore the client will reject the forged state $\mathcal{B}$ almost surely.
\end{proof}

\subsection{Sealing Entropy Flow from Secrets}
ESTMF requires that the public state does not leak information about the secrets it protects. This is achieved by using secure cryptographic primitives. The following proofs establish that MFKDF2's use of a Pseudorandom Permutation (PRP) prevents the key recovery and two-time pad vulnerabilities found in NS23.

\begin{theorem}
\label{thm:prp_kpa_security}
    If a block cipher $E$ is a secure Pseudorandom Permutation (PRP), then it is secure against key recovery from known plaintext-ciphertext pairs.
\end{theorem}
\begin{proof}
    This is a standard result in symmetric-key cryptography \cite[Ch. 3]{BonehShoup2020}. If an adversary could recover the key from a known plaintext-ciphertext pair with non-negligible advantage, they could be used as a subroutine to distinguish the PRP from a truly random permutation, thus breaking its fundamental security definition.
\end{proof}

\begin{theorem}
\label{thm:prp_share_encryption}
    Let the MFKDF2 thresholding scheme use a secure PRP $E$ to encrypt a share $s$ (or factor secret $\sigma$, resp.). An adversary observing two ciphertexts of the same share under different share keys, $c_{s} = E_{k}(s)$ (resp. $c_{\sigma} = E_{k}(\sigma)$) and $c'_{s} = E_{k'}(s)$ (resp. $c_{\sigma} = E_{k}(\sigma)$), gains no advantage in recovering $s$ (resp. $\sigma$), $k$, or $k'$.
\end{theorem}
\begin{proof}
    Let the share keys $k$ and $k'$ be derived from independent computational $m$-entropy factors. An adversary is given the two ciphertexts $c_{s}$ and $c'_{s}$, which without loss of generality may be encrypted factor shares $s$ or factor secrets $\sigma$. Since $E$ is a secure PRP, its outputs are computationally indistinguishable from random permutations, revealing no information about the plaintext $s_i$ or the keys $k$ or $k'$.

    Therefore, the adversary's best strategy is an exhaustive search. To find the secrets, the adversary must iterate through all $2^m \times 2^m = 2^{2m}$ possible pairs of keys $(k^*, k'^*)$ and check if the decryptions match: $E^{-1}_{k^*}(c_{s}) \stackrel{?}{=} E^{-1}_{k'^*}(c'_{s})$. The $\mathcal{O}(2^{2m})$ complexity of this search is computationally infeasible for any factor with non-trivial entropy. Thus, the adversary's advantage is negligible.
\end{proof}

\subsection{Sealing Entropy Flow Between Factors}
ESTMF requires that factors be cryptographically isolated to prevent collateral damage from a partial compromise. MFKDF2 achieves this through a combination of protocol logic and cryptographic separation.

\begin{theorem}
\label{thm:share_updating}
    Let the \texttt{Recover} operation in an MFKDF thresholding scheme be such that it generates a fresh set of shares $\{s'_i\}$ for the new state $\mathcal{B}'$. An adversary who compromises the source key material of an old factor $F_i$ after recovery gains negligible advantage in recovering any valid share $s'_i$ from the current state $\mathcal{B}'$.
\end{theorem}
\begin{proof}
    Let an adversary compromise the old factor $F_i$ and thus learn its share key $k_i$. The adversary's knowledge of $k_i$ allows for the recovery of the old share $s_i$ from the old state $\mathcal{B}$.

    Since \texttt{Recover} operation generates a new set of shares $\{s'_i\}$ from a new, randomly chosen secret-sharing polynomial. Since the polynomial is randomly generated from a secure PRNG, the old share $s_i$ is computationally independent of any new share $s'_i$ in the current state $\mathcal{B}'$. Therefore, the adversary's knowledge of the compromised factor's secrets ($k_i$ and $s_i$) provides no advantage in determining any of the new shares $\{s'_i\}$, and the attacker has only an advantage in key derivation for the already partially-compromised state $\mathcal{B}$.
\end{proof}

\begin{theorem}
\label{thm:factor_independence_n_of_n}
    The MFKDF2 n-of-n construction, instantiated with factors satisfying the Factor-KI property is Factor-IND-CMA secure.
\end{theorem}
\begin{proof}
    We prove that for any PPT adversary $\mathcal{A}$, its advantage in the Factor-IND-CMA (Def. \ref{def:factor_independence_game}) game is negligible. First, we note that any attempt by the adversary to provide a tampered or replayed state will be rejected by the client's self-referential MAC check (Prop. \ref{prop:state_integrity}). Thus, the adversary is reduced to a passive observer of a valid state stream.

    The adversary's goal is to distinguish between two worlds: one using challenge factor $F_A$, the other $F_B$. The only difference in the adversary's view between these two worlds is the public state of the challenge factor, $\beta_{A}$ vs. $\beta_{B}$. The master secret $\kappa$ is formed by concatenating the salted hashes of all factor source keys. Since concatenation is non-commutative and each factor is salted, the contribution of the challenge factor is bound to its specific position in the policy. Since the factors are assumed to be Factor-KI secure, the adversary's view is computationally indistinguishable in both worlds. Therefore, its advantage is negligible.
\end{proof}

\begin{lemma}
\label{lem:sss_indistinguishability}
    Let a $b$-bit secret be shared using a $(t, n)$ Shamir's Secret Sharing (SSS) scheme performed bytewise over the finite field $F = \operatorname{GF}(2^q)$, where $b$ is a multiple of $q$ and $n < 2^q$. Any resulting share $s_i$ is computationally indistinguishable from a uniformly random $b$-bit string.
\end{lemma}
\begin{proof}
    We prove this by showing that the statistical distribution of a share is identical to the uniform distribution over $\{0,1\}^b$. An adversary's advantage in any distinguishing game is therefore zero.

    The SSS construction operates on each $q$-bit element $S_k \in F$ of the secret independently by constructing a random polynomial $f_k(x)$ of degree $t-1$ such that $f_k(0) = S_k$. The remaining $t-1$ coefficients are chosen independently and uniformly at random from $F$. A share-element $s_{i,k}$ is the evaluation $f_k(x_i)$ at a distinct, non-zero public point $x_i$. Since at least one coefficient of $f_k(x)$ is uniformly random (for any $t>1$), the value of $f_k(x_i)$ is a uniformly random element of $F$.

    The final $i$-th share, $s_i$, is the concatenation of these independent and uniformly random $q$-bit share-elements: $s_i = s_{i,1} \ || \ \dots \ || \ s_{i,k}$. The concatenation of $k$ independent, uniformly random $q$-bit strings is itself a uniformly random $b$-bit string. Since the distribution of a real share is identical to the uniform distribution, no distinguisher can gain an advantage. 
\end{proof}

\begin{theorem}
\label{thm:factor_independence_threshold}
    The MFKDF2 threshold construction, instantiated with factors satisfying the Factor-KI property is Factor-IND-CMA secure.
\end{theorem}
\begin{proof}
    Following Theorem \ref{thm:factor_independence_n_of_n}, we argue the only difference now manifests in its encrypted share ($c_{A}$ vs. $c_{B}$) and its public factor state ($\beta_{A}$ vs. $\beta_{B}$). Assuming 256 bit keys, Lemma \ref{lem:sss_indistinguishability} guarantees the shares $s_A$ and $s_B$ are random and unbiased when doing SSS over $\operatorname{GF}(2^q)$. Assuming independent factors, the share keys $k_A$ and $k_B$ are also independent. As the encryption scheme $E$ is a secure PRP, the resulting ciphertexts $c_{A} = E(k_A, s_A)$ and $c_{B} = E(k_B, s_B)$ are computationally indistinguishable from random strings.

    Similarly, any secret material within the factor states $\beta_A$ and $\beta_B$ is also protected by a secure PRP, rendering their state streams computationally indistinguishable. Since the adversary's entire view is indistinguishable, it cannot guess the challenge bit $b$ with a probability significantly greater than $1/2$, and its advantage is therefore negligible.
\end{proof}

\subsection{Overall System Security (MSI)}
Finally, we combine our previous results to prove that MFKDF2 satisfies the highest-level ESTMF property: Master Secret Indistinguishability.

\begin{theorem}
\label{thm:mfkdf2_achieves_msi}
    The MFKDF2 scheme achieves Master Secret Indistinguishability, assuming its underlying primitives (KDF, PRP, MAC) are secure and its factor constructions are factor-KI secure.
\end{theorem}
\begin{proof}
    We prove this by showing that the adversary's view of the public state stream is computationally indistinguishable in the two worlds of the MSI game ($M_A$ vs. $M_B$). The proof proceeds via a standard hybrid argument over the components of the public state stream $((\mathcal{B}_1, \mathcal{T}_1), (\mathcal{B}_2, \mathcal{T}_2), \dots, (\mathcal{B}_q, \mathcal{T}_q))$. If each component of the state stream is indistinguishable, then the entire stream is indistinguishable.

    First, the integrity of all public parameters, including the policy and salts, is guaranteed by the self-referential MAC (Prop. \ref{prop:state_integrity}). These parameters are independent of the master secret and thus provide no advantage to the adversary.

    Second, the stream of encrypted shares $\{\{c_{i}\}_1, \{c_i\}_2, \dots, \{c_i\}_q\}$ is indistinguishable. The shares derived from the independent master secrets $M_A$ and $M_B$ are themselves independent (follows from Lemma \ref{lem:sss_indistinguishability}). As these shares are encrypted with a secure PRP (Thm. \ref{thm:prp_share_encryption}) and regenerated upon recovery (Thm. \ref{thm:share_updating}), the resulting stream of ciphertexts is computationally indistinguishable from random.

    Finally, the individual factor states $\{\beta_i\}$ and the MAC tag $T$ are also indistinguishable. The final keys $K_A$ and $K_B$ are derived from the master secrets via a KDF (modeled as a PRF), making them indistinguishable from random. By the factor-KI (Def. \ref{def:leak_free_factor}) assumption on the factor constructions and the PRF property of the MAC, the factor states and MAC tags generated using these keys are also computationally indistinguishable.

    Since all components of the public state stream are computationally indistinguishable whether they originate from master secret $M_A$ or $M_B$, the adversary's advantage in the MSI game is negligible.
\end{proof}


\clearpage
\newpage

\end{document}